\documentclass{LMCS}

\def\dOi{10(1:2)2014}
\lmcsheading%
{\dOi}
{1--36}
{}
{}
{Dec.~\phantom09, 2011}
{Jan.~21, 2014}
{}

\subjclass{G.2.2 Graph Theory---Hypergraphs, F.4.1 Mathematical
  Logic---Model Theory.}  

\ACMCCS{[{\bf Mathematics of computing}]: Discrete mathematics---Graph
  theory---Hypergraphs; [{\bf Theory of computation}]: Logic---Finite
  Model Theory}

\usepackage[british]{babel}
\usepackage[T1]{fontenc}
\usepackage{microtype,amssymb,mathrsfs,relsize,graphics,hyperref}

\newcommand*{\MSO}{\mathrm{MSO}}

\newcommand*{\CMSO}{\mathrm{CMSO}}
\newcommand*{\edg}{\mathrm{edg}}

\newcommand*{\cwd}{\qopname\relax o{cwd}}

\newcommand*{\MTh}{\mathrm{MTh}}

\newcommand*{\qr}{\mathrm{qr}}
\newcommand*{\WD}{\mathrm{wd}}
\newcommand*{\SWD}{\mathrm{swd}}
\newcommand*{\lex}{\mathrm{lex}}
\newcommand*{\Pred}{\mathrm{Pred}}

\newcommand*{\Sep}{\mathrm{Sep}}
\newcommand*{\SEP}{\mathsf{SEP}}
\newcommand*{\Cut}{\mathrm{Cut}}
\newcommand*{\CUT}{\mathsf{CUT}}
\newcommand*{\Del}{\mathrm{Del}}
\newcommand*{\Inc}{\mathrm{Inc}}
\newcommand*{\IS}{\mathrm{IS}}
\newcommand*{\emptyseq}{\langle\rangle}
\newcommand*{\?}{\kern .08em}
\newcommand*{\nmodels}{\not\models}
\newcommand*{\nsqsubseteq}{\not\sqsubseteq}
\let\lso\rightarrow
\let\liff\leftrightarrow
\let\Land\bigwedge
\let\Lor\bigvee
\newcommand\PSet{\mathcal{P}}

\newcommand\frakA{\mathfrak{A}}
\newcommand\frakB{\mathfrak{B}}
\newcommand\frakC{\mathfrak{C}}
\newcommand\calC{\mathcal{C}}
\newcommand\calK{\mathcal{K}}
\newcommand\calP{\mathcal{P}}
\newcommand\calT{\mathcal{T}}
\newcommand\bbN{\mathbb{N}}

\newcommand*{\lset}{\{\,}
\newcommand*{\mset}{\mathrel|}
\newcommand*{\rset}{\,\}}
\newcommand*{\biglset}{\bigl\{\,}
\newcommand*{\bigmset}{\bigm|}
\newcommand*{\bigrset}{\,\bigr\}}
\newcommand*{\abs}[1]{\lvert#1\rvert}
\newcommand*{\bigabs}[1]{\bigl\lvert#1\bigr\rvert}
\newcommand*{\set}[2]{\lset#1\mset#2\rset}
\newcommand*{\bigset}[2]{\biglset#1\bigmset#2\bigrset}
\newcommand*{\qtextq}[1]{\quad\text{#1}\quad}

\newcommand*\defiff{\mathrel{{:}{\Longleftrightarrow}}}
\renewcommand*\iff{\Longleftrightarrow}

\newcommand\closingmark{\hfill\hbox{}\hspace*{0pt plus 1fill}$\Diamond$}
\newcommand\upqed{\vskip-\baselineskip\vskip-\belowdisplayskip}

\newcommand\prefixtext[1]{%
  \ifvmode\else\\\@empty\fi
  \noalign{%
    \penalty0%
    \vbox{\mathstrut}%
    \penalty10000%
    \vskip-\baselineskip
    \penalty10000%
    \vbox to 0pt{%
      \normalbaselines
      \ifdim\linewidth=\columnwidth
      \else
        \parshape\@ne
        \@totalleftmargin\linewidth
      \fi
      \vss
      \noindent#1\par}%
      \penalty10000%
      \vskip-\baselineskip}%
      \penalty10000}

\begin{document}
\title[Monadic second-order definable graph orderings]
  {Monadic second-order definable graph orderings}
\author[A.~Blumensath]{Achim Blumensath\rsuper a}
\address{{\lsuper a}TU Darmstadt}
\email{blumensath@mathematik.tu-darmstadt.de}
\thanks{{\lsuper a}Work partially supported by DFG grant BL~1127/2-1.}
\author[B.~Courcelle]{Bruno Courcelle\rsuper b}
\address{{\lsuper b}Labri, Bordeaux University, Honorary member of Institut Universitaire
  de France}
\email{courcell@labri.fr}

\keywords{Monadic second-order logic, Definability, Linear orders}

\begin{abstract}
We study the question of whether, for a given class of finite graphs,
one can define, for each graph of the class,
a linear ordering in monadic second-order logic,
possibly with the help of monadic parameters.
We consider two variants of monadic second-order logic\?: one where
we can only quantify over sets of vertices and one where we can also
quantify over sets of edges.
For several special cases, we present combinatorial characterisations of when
such a linear ordering is definable.
In some cases, for instance for graph classes that omit a fixed graph as a
minor, the presented conditions are necessary and sufficient\?;
in other cases, they are only necessary.
Other graph classes we consider include
complete bipartite graphs, split graphs, chordal graphs, and cographs.
We prove that orderability is decidable for the so called HR-equational
classes of graphs, which are described by equation systems and
generalize the context-free languages.
\end{abstract}

\maketitle

\section{Introduction}

When studying the expressive power of monadic second-order logic ($\MSO$) for
finite graphs, often the question arises of whether one can define a linear
order on the vertex set.
For instance, the property that a set has even cardinality cannot, in general,
be expressed in $\MSO$. If, however, the considered set is linearly ordered,
we can write down a corresponding $\MSO$-formula.
The same holds for every predicate $\mathrm{Card}_q(X)$ expressing that
the cardinality of the set~$X$ is a multiple of~$q$.
It follows that the extension of $\MSO$
by the predicates $\mathrm{Card}_q(X)$,
called \emph{counting monadic second-order logic} ($\CMSO$),
is no more powerful than $\MSO$ on every class of structures
on which a linear order is $\MSO$-definable.

Another example of a situation where the availability of a linear order
facilitates certain logical constructions
is the definability of graph decompositions
such as the modular decomposition of a graph.
It is shown in~\cite{Courcelle96b} that the modular decomposition
of a graph is definable in $\MSO$ if the graph is equipped with a linear order.
Finally,\footnote{Yet another example is the construction of (a
combinatorial description of) a plane embedding of a connected planar graph.
Such embeddings are definable in $\MSO$
if we can order the neighbours of each vertex (see \cite{Courcelle00}).
For $3$-connected graphs such an ordering is always definable,
but for graphs that are not $3$-connected this is not always the case.}
although we will not address complexity questions in this article,
we recall that, over linearly ordered structures, the
complexity class PTIME is captured by
least fixed-point logic~\cite{EbbinghausFlum95,Immerman99}.

A~formula $\varphi(x,y)$ with two free first-order variables $x$~and~$y$
\emph{defines a (linear) order} on a relational structure~$\frakA$
if the binary relation consisting of all pairs $(a,b)$ of elements of~$\frakA$
satisfying $\frakA \models \varphi(a,b)$ is a linear order on~$A$.
We say that $\varphi(x,y)$ \emph{defines an order on a class of structures}
if it defines a linear order on each structure of that class.
Our objective is to provide combinatorial characterisations of classes
of finite graphs whose representing structures are \emph{$\MSO$-orderable,}
i.e., on which one can define an order by an $\MSO$-formula.
(The question of whether a \emph{partial} order is definable is trivial
since equality is a partial order.
Therefore, we only consider linear orders in this article.)

As defined above the notion of an $\MSO$-orderable class is too restrictive.
To get interesting results, we allow in the above definitions formulae with
\emph{parameters.}
That is, we take a formula $\varphi(x,y;\bar Z)$ with additional free
set variables $\bar Z = \langle Z_0,\dots,Z_{n-1}\rangle$ and,
for each structure~$\frakA$ in the given class, we choose values
$P_0,\dots,P_{n-1} \subseteq A$ for these variables such that the binary
relation
\begin{align*}
  \set{ (a,b) }{ \frakA \models \varphi(a,b;\bar P) }
\end{align*}
is a linear order on~$A$.

There is no $\MSO$-formula (even with parameters) that defines a linear order
on all finite graphs.
An easy way to see this is to observe that every ordered structure
is \emph{rigid,} i.e., that it has no non-trivial automorphism.
Since we can find graphs that are not rigid,
even after labelling them with a fixed number of parameters,
it follows that no formula can order all graphs.
The same argument shows that the following classes of finite graphs
are not $\MSO$-orderable\?:
(1) graphs without edges\?;
(2) cliques\?;
(3) stars\?;
(4) trees of a fixed height\?; and
(5) bipartite graphs.
On the other hand, to take an easy example,
the class of all finite connected graphs
of degree at most~$d$ (for fixed~$d$) is $\MSO$-orderable. 

If graphs are replaced by their incidence graphs, $\MSO$-formulae become more
powerful, because they can use quantifications over sets of edges.
In this case we speak of $\MSO_2$-orderable classes.
Otherwise, we call the class $\MSO_1$-order\-able.
Due to the greater expressive power,
the family of $\MSO_2$-orderable classes properly includes that of
$\MSO_1$-orderable ones.
This means that, in the combinatorial characterisations presented below,
the conditions for $\MSO_1$-order\-ability must be stronger than those for
$\MSO_2$-orderability.
For instance, the class of all cliques is $\MSO_2$-orderable
but not $\MSO_1$-orderable.

There are simple combinatorial criteria showing that a class is not
$\MSO$-order\-able. For instance, a class of trees is not $\MSO$-orderable
if the degree of vertices is unbounded. The reason is that an
$\MSO$-formula can only distinguish between a bounded number of neighbours
of a vertex. If the number of neighbours is too large, we can swap
two of the attached subtrees without affecting the truth value of the formula.
Generalising this example, we obtain the following criterion for
$\MSO_2$-orderability\?: if a class~$\calC$ is $\MSO_2$-orderable,
there exists a function~$f$ such that, whenever we remove~$k$ vertices
from a graph in~$\calC$, the resulting graph has at most $f(k)$
connected components (Proposition~\ref{Prop: Sep boundend}).

In many cases, it turns out that this necessary condition is also sufficient.
For instance, we will prove in
Theorem~\ref{Thm: GSO-orderable with excluded minor}
below that a class of graphs omitting some graph as a minor
is $\MSO_2$-orderable if, and only if, it has the above property.

This article is organised as follows.
Sections \ref{Sect: preliminaries}~and~\ref{Sect: general orderability}
introduce notation and basic definitions.
The main part consists of Sections \ref{Sect: MSO2}~and~\ref{Sect: MSO},
which collect our results on, respectively,
$\MSO_2$-orderability and $\MSO_1$-orderability.

For $\MSO_2$-orderability, we present a necessary condition
in Section~\ref{SSect: MSO2-necessary}.
We prove that this condition is also sufficient for
trees (Theorem~\ref{Thm: orderability of trees}) and, more generally,
for classes of graphs omitting some graph as a minor
(Theorem~\ref{Thm: GSO-orderable with excluded minor}).
For some classes of bipartite graphs and of split graphs,
we obtain a similar result,
using a slightly stronger combinatorial condition
(Theorems \ref{Thm: GSO-orderable d-partite graphs}~and~%
\ref{Thm: orderability of split graphs}).
Furthermore, we prove that some classes are not $\MSO_2$-orderable
in a very strong sense\?: they contain no infinite subclass that
is $\MSO_2$-orderable. This is the case for
trees of bounded height
(Corollary~\ref{Cor: trees of bounded depth hereditarily MSO2-unorderable})
and graphs of bounded $n$-depth tree-width
(Proposition~\ref{Prop: n-depth tree-width hereditarily MSO2-unorderable}).
Finally, we also prove that, for certain effectively presented classes of
graphs, $\MSO_2$-orderability is decidable
(Corollary~\ref{Cor: orderability decidable}).

For $\MSO_1$-orderability the picture we obtain is slightly more sketchy.
We present a necessary condition for $\MSO_1$-orderability
in Section~\ref{Sect: MSO orders}.
We prove that it is also sufficient for cographs
(Theorem~\ref{Thm: cographs}) and
graphs of bounded $n$-depth $\otimes$-width
(Theorem~\ref{Thm: characterisation of orderable graphs of bounded otimes-width}).

Finally, we consider reductions between orderability properties
in Section~\ref{Sect: reductions}.
We show that, for split graphs and bipartite graphs,
the question of $\MSO_i$-orderability is as hard as for arbitrary graphs.
This indicates that we are far from having a
combinatorial characterisation of orderability for such classes.

\section{Preliminaries}   
\label{Sect: preliminaries}

Let us fix our notation and terminology.
We write $[n] := \{0,\dots,n-1\}$, for $n \in \bbN$.
We denote tuples $\bar a = \langle a_0,\dots,a_{n-1}\rangle$ with a bar.
The empty tuple is $\emptyseq$.
We write $A \mathbin\Delta B$ for the symmetric difference of
two sets $A$~and~$B$.
We denote partial orders by symbols like $\leq$~and~$\preceq$,
and the corresponding strict partial orders by $<$~and~$\prec$, respectively.

\subsection{Structures and graphs}
\label{Sect: structures and graphs}

In this article we consider only purely relational structures
$\frakA = \langle A,R^\frakA_0,\dots,R^\frakA_{n-1}\rangle$
with finite signatures $\Sigma = \{R_0,\dots,R_{n-1}\}$.
The universe~$A$ will always be finite, and we allow it to be empty
as this convention is common in graph theory.
In some places we will also allow relational structures with constants,
but when doing so we will always mention it explicitly.
For a relation~$R$ and a set~$X$, we write $R \restriction X$ for the restriction
of~$R$ to~$X$. For a tuple $\bar R$ of relations, we denote by $\bar R \restriction X$
the corresponding tuple of restrictions.

We will mainly consider graphs instead of arbitrary relational
structures. For basic notions of graph theory, we refer the reader
to the book~\cite{Diestel10}.
In this article, graphs will always be finite, simple, loop-free, and
undirected, with the exception of rooted trees and forests,
which we consider to be oriented (see below).
We will denote the edge between vertices $u$~and~$v$ by $(u,v)$.
Note that the same edge can also be written as $(v,u)$.
There are two ways to represent a graph $G = \langle V,E\rangle$ by a structure.
Both of them will be used.
We can use structures of the form $\lfloor G\rfloor := \langle V,\edg\rangle$
where the universe~$V$ consists of the set of vertices and we have a binary
edge relation $\edg \subseteq V \times V$,
or we can use structures of the form
$\lceil G\rceil := \langle V \cup E,\mathrm{inc}\rangle$
where the universe contains both, the vertices and the (undirected) edges
of the graph and we have a binary incidence relation
$\mathrm{inc} \subseteq V \times E$ telling us which vertices belong to which
edges.
If $\calC$~is a class of graphs, we denote the corresponding classes of
relational structures by $\lfloor\calC\rfloor$ and $\lceil\calC\rceil$,
respectively.

Forests will always be rooted and directed in such a way that
every edge is oriented away from the root.
The \emph{tree-order} associated with a forest~$F$
is the partial order defined by
\begin{align*}
  x \preceq_F y \quad\defiff\quad
  \text{some path from a root to $y$ contains } x\,.
\end{align*}
If $x \prec y$, we call $x$~a \emph{predecessor} of~$y$ and
$y$~a \emph{successor} of~$x$.
We speak of \emph{immediate predecessors} and \emph{immediate successors}
if there is no vertex in between.
The \emph{$n$-th level} of a forest~$F$ consists of all vertices at
distance~$n$ from some root. Hence, the roots form level~$0$.
The \emph{height} of~$F$ is the maximal level of its vertices.

\begin{defi}\label{Def: sparse}
A graph $G = \langle V,E\rangle$ is \emph{$r$-sparse}\footnote{In~\cite{CourcelleEngelfriet12}
such graphs are called \emph{uniformly $r$-sparse.}}
if, for every subset $X \subseteq V$, we have
$\bigabs{E \restriction X} \leq r\cdot\abs{X}$.
\closingmark\end{defi}

We denote by $\frakA \oplus \frakB$ the disjoint union of the structures
$\frakA$~and~$\frakB$.
For structures $\lfloor G\rfloor$~and~$\lfloor H\rfloor$ encoding graphs,
we also use a dual operation $\lfloor G\rfloor \otimes \lfloor H\rfloor$ that,
after forming the disjoint union of $\lfloor G\rfloor$~and~$\lfloor H\rfloor$,
adds all possible edges connecting a vertex of~$G$ to a vertex of~$H$.
For a set $S \subseteq A$ of elements, we write
$\frakA[S]$ for the substructure of~$\frakA$ induced by~$S$
and $\frakA - S$ for $\frakA[A - S]$.
We use the analogous notation $G[S]$ and $G - S$, for graphs~$G$.

We assume that the reader is familiar with the notion of a tree decomposition
and the tree-width of a graph
(see, e.g., \cite{Diestel10,CourcelleEngelfriet12}).
At a few places, we will refer to a variant of tree-width, called
\emph{$n$-depth tree-width,} that was introduced
in~\cite{BlumensathCourcelle10}.
It is defined in terms of tree decompositions where the height of the
index tree is at most~$n$.

Finally, we will employ tools related to the notion of \emph{clique-width,}
which is defined for graphs with \emph{ports} in a finite set~$[k]$,
that is, graphs $G = \langle V,E,\pi\rangle$ equipped with a function
$\pi : V \to [k]$.
We say that a vertex $a \in V$ \emph{has port label~$a$} if $\pi(v) = a$.
The notion of clique-width is defined in terms of the following operations
on graphs\footnote{For a detailed discussion of concrete graphs
versus graphs defined up to isomorphism, see Section~2.2
of~\cite{CourcelleEngelfriet12}}
with ports\?:
\begin{itemize}
\item for each $a \in [k]$, a constant~$a$ denoting the graph with
  a single vertex that has port label~$a$\?;
\item the disjoint union~$\oplus$ of two graphs with ports\?;
\item the edge addition operation $\mathrm{add}_{a,b}$, for $a,b \in [k]$,
  adding all edges between some vertex with port label~$a$
  and some vertex with port label~$b$ that do not already exist\?;
\item the port relabelling operation $\mathrm{relab}_h$, for $h : [k] \to [k]$,
  changing each port label~$a$ to the port label $h(a)$.
\end{itemize}
Each term using these operations defines a graph with ports in~$[k]$.
The clique-width of a graph $G = \langle V,E\rangle$ is the least number~$k$
such that, for some function $\pi : V \to [k]$,
there exists a term denoting $\langle G,\pi\rangle$
(for details cf.~\cite{CourcelleEngelfriet12,CourcelleEngelfrietRozenberg93,CourcelleOlariu00}).
We denote the clique width of~$G$ by $\cwd(G)$.

\subsection{Monadic second-order logic}

Monadic second-order logic ($\MSO$)\footnote{There is also
counting monadic second-order logic ($\CMSO$) which extends $\MSO$
by set predicates of the form $\mathrm{Card}_q(X)$ expressing that
the cardinality of~$X$ is a multiple of~$q$.
Although our results are stated and proved for $\MSO$,
they also hold for $\CMSO$\?: the technical core of our proofs
is the composition theorem which holds for $\CMSO$ as well.
We currently do not have an example of a class of structures that is
$\CMSO$-orderable but not $\MSO$-orderable, but it seems likely that such
classes do exist.}
is the extension of first-order logic
by set variables and quantifiers over such variables.
The \emph{quantifier-rank} $\qr(\varphi)$ of an $\MSO$-formula~$\varphi$ is the maximal number
of nested quantifiers in~$\varphi$, where we count both, first-order and second-order quantifiers.
The \emph{monadic second-order theory} of quantifier rank~$h$ of a structure~$\frakA$ is
the set of all $\MSO$-formulae of quantifier rank~$h$ satisfied by~$\frakA$.
We denote it by $\MTh_h(\frakA)$.
Frequently, we are interested not in the theory of the structure~$\frakA$ itself,
but in the theory of an expansion $\langle\frakA,\bar P,\bar a\rangle$ by unary
predicates~$\bar P$ and constants~$\bar a$. In this case we write
$\MTh_h(\frakA,\bar P,\bar a)$ omitting the brackets.
Note that such situations are the only ones in which
we allow constants in structures.

Let us remark that, for a fixed signature and a given maximal quantifier-rank,
there are only finitely many formulae up to logical equivalence.
Furthermore, we can effectively compute an upper bound on the number of
classes and there exists an effective normal form for formulae.
However, since equivalence of formulae is undecidable,
this normal form does not represent logical equivalence.
Hence, some equivalence classes contain several formulae in normal form.
Details can be found, e.g., in Section~5.6 of~\cite{CourcelleEngelfriet12}.
In particular, it follows that, for every $h \in \bbN$,
there are only finitely many theories of quantifier-rank~$h$
and we can represent each such theory by the finite set of formulae
in normal form it contains.
A~detailed calculation shows that the number of such theories is
roughly $\exp_h(n)$ where
\begin{align*}
  \exp_0(n) := n
  \qtextq{and}
  \exp_{k+1}(n) := 2^{\exp_k(n)}
\end{align*}
and the number~$n$ only depends on the signature,
but not on the quantifier-rank~$h$.
Recall that a function $f : \bbN \to \bbN$ is \emph{elementary} if it is
bounded from above by a function of the form $\exp_k$, for some
fixed~$k \in \bbN$.
Furthermore, it follows that we can construct, for each theory~$\Theta$ of
quantifier-rank~$h$, a single formula~$\chi_\Theta$ that is equivalent to it,
i.e., such that
\begin{align*}
  \frakA \models \chi_\Theta \quad\iff\quad
  \MTh_h(\frakA) = \Theta\,.
\end{align*}
In fact, $\chi_\Theta$~is just the conjunction of all formulae in normal
form contained in~$\Theta$. For this reason we will also denote it by
$\Land\Theta$.

Let $\varphi(\bar x,\bar Y;\bar Z)$ be an $\MSO$-formula
with free first-order variables~$\bar x$ and free second-order variables $\bar Y,\bar Z$.
Given a structure~$\frakA$ and sets $P_i \subseteq A$,
we can assign the values~$\bar P$ to the variables~$\bar Z$.
This way we obtain a formula $\varphi(\bar x,\bar Y;\bar P)$ with partially assigned variables.
The values $\bar P$ are called the \emph{parameters} of this formula.
The \emph{relation defined} by a formula $\varphi(\bar x;\bar P)$ in a structure~$\frakA$
is the set
\begin{align*}
  \varphi(\bar x;\bar P)^\frakA := \set{ \bar a }{ \frakA \models \varphi(\bar a;\bar P) }\,.
\end{align*}

One important tool to compute monadic theories is the so-called Composition
Theorem (see, e.g,
\cite{Makowsky04,BlumensathColcombetLoeding07,CourcelleEngelfriet12}),
which allows one to compute the theory of a structure composed from smaller
parts from the theories of these parts.
There are several variants of the Composition Theorem.
We will employ the following version.
\begin{defi}
Let $\frakA_0,\dots,\frakA_{m-1}$ be structures and let
$\bar a^i = \langle a^i_0,\dots,a^i_{n-1}\rangle \in A_i^n$ be $n$-tuples,
for $i < m$.
The \emph{amalgamation} of the structures~$\frakA_i$ over the parameters~$\bar a^i$
is the structure~$\langle\frakA',\bar a'\rangle$ obtained from the disjoint union
$\frakA_0 \oplus\dots\oplus \frakA_{m-1}$ by,
for every $k < n$, merging the elements $a^0_k,\dots,a^{m-1}_k$ into a single element~$a'_k$.
The tuple $\bar a' = \langle a'_0,\dots,a'_{n-1}\rangle$ consists
of the elements resulting from the merging.
\closingmark\end{defi}
\begin{thm}[Composition Theorem]\label{Thm: composition}
Let $\frakA_0,\dots,\frakA_{m-1}$ be structures
and, for $i < m$, let $\bar a_i \in A_i^n$ be $n$-tuples
and $\bar c_i \in A_i^{l_i}$ $l_i$-tuples.
Let $\langle\frakA',\bar a'\rangle$ be the amalgamation of
the structures~$\frakA_i$ over~$\bar a_i$.
Then
\begin{align*}
  \MTh_h(\frakA',\bar a'\bar c_0\dots\bar c_{m-1})
\end{align*}
is uniquely determined by the theories
\begin{align*}
  \MTh_h(\frakA_0,\bar a_0\bar c_0),\dots,
  \MTh_h(\frakA_{m-1},\bar a_{m-1}\bar c_{m-1})\,.
\end{align*}
Furthermore, the function mapping these theories to the theory of the
amalgamation is computable.
\end{thm}
Since disjoint unions are particular amalgamations, we obtain
the following corollary.
\begin{cor}\label{Cor: composition for disjoint union}
There exists an computable function mapping
$\MTh_h(\frakA)$ and $\MTh_h(\frakB)$ to $\MTh_h(\frakA \oplus \frakB)$.
\end{cor}

\subsection{Transductions}

The notion of a monadic second-order transduction provides
a versatile framework to define transformations of structures.
To simplify the definition we first introduce three particular types
of transductions
and we obtain $\MSO$-transductions as compositions of these.
\begin{defi}\label{Def: transduction}
(a) Let $k \geq 2$ be a natural number.
The operation $\mathrm{copy}_k$ maps a structure~$\frakA$ to the expansion
\begin{align*}
  \mathrm{copy}_k(\frakA) :=
    \langle \frakA\oplus\dots\oplus\frakA, {\sim}, P_0,\dots,P_{k-1}\rangle
\end{align*}
of the disjoint union of $k$~copies of~$\frakA$ by the following relations.
Denoting the copy of an element $a \in A$ in the $i$-th component of
$\frakA \oplus\dots\oplus \frakA$ by the pair $\langle a,i\rangle$, we define
\begin{align*}
  P_i := \set{ \langle a,i\rangle }{ a \in A }
  \qtextq{and}
  \langle a,i\rangle \sim \langle b,j\rangle \ \defiff\ a = b\,.
\end{align*}
For $k = 1$, we set $\mathrm{copy}_1(\frakA) := \frakA$.

(b) For $m \in \bbN$, we define the multi-valued operation $\mathrm{exp}_m$
that maps a structure~$\frakA$ to all of its possible expansions by
$m$~unary predicates $Q_0,\dots,Q_{m-1} \subseteq A$.
Note that $\mathrm{exp}_0$ is just the identity.

(c) A \emph{basic $\MSO$-transduction} is a partial operation~$\tau$
on relational structures described by a list
\begin{align*}
  \bigl\langle \chi, \delta(x),
               \varphi_0(\bar x),\dots,\varphi_{s-1}(\bar x)\bigr\rangle
\end{align*}
of $\MSO$-formulae called the \emph{definition scheme} of~$\tau$.
Given a structure~$\frakA$ that satisfies the sentence~$\chi$,
the operation~$\tau$ produces the structure
\begin{align*}
  \tau(\frakA) := \langle D, R_0,\dots,R_{s-1}\rangle
\end{align*}
where
\begin{align*}
  D := \set{ a \in A }{ \frakA \models \delta(a) }
  \qtextq{and}
  R_i := \set{ \bar a \in D^{\varrho_i} }
             { \frakA \models \varphi_i(\bar a) }\,.
\end{align*}
($\varrho_i$~is the arity of~$R_i$.)
If $\frakA \nmodels \chi$ then $\tau(\frakA)$ is undefined.

(d) A \emph{quantifier-free transduction} is a basic $\MSO$-transduction,
where all formulae in the definition scheme are quantifier free.

(e) A \emph{$k$-copying $\MSO$-transduction}~$\tau$ is a
(multi-valued) operation on relational structures of the form
$\tau_0 \circ \mathrm{copy}_k \circ \mathrm{exp}_m$
where $\tau_0$~is a basic $\MSO$-trans\-duc\-tion.
When the value of~$k$ does not matter, we will simply speak of
a \emph{transduction.}

Due to~$\mathrm{exp}_m$, a structure can be mapped
to several structures by~$\tau$.
Consequently, we define $\tau(\frakA)$ as the \emph{set} of possible values
$(\tau_0 \circ \mathrm{copy}_k)(\frakA, \bar P)$ where $\bar P$~ranges
over all $m$-tuples of subsets of~$A$.

(f) An $\MSO$-transduction~$\tau$ is \emph{domain-preserving} if,
it is $1$-copying and,
for every structure~$\frakA$ such that $\tau(\frakA)$ is defined,
the image~$\tau(\frakA)$ has the same universe as~$\frakA$.
\closingmark\end{defi}

\begin{rem}
(a) The expansion by $m$~unary predicates corresponds, in the terminology
of \cite{Courcelle95,Courcelle03}, to using $m$~\emph{parameters.}

(b) Note that every basic $\MSO$-transduction
is a $1$-copying $\MSO$-transduction without parameters.
\closingmark\end{rem}

The most important property of $\MSO$-transductions is the fact
that they are compatible with $\MSO$-theories in the following sense
(see, e.g., Theorem~5.10 of~\cite{CourcelleEngelfriet12}).
\begin{lem}[Backwards Translation]\label{Lem: transductions are comorphisms}
Let $\tau$~be a transduction. For every $\MSO$-sentence~$\varphi$, there
exists an $\MSO$-sentence~$\varphi^\tau$ such that, for all structures~$\frakA$,
\begin{align*}
  \frakA \models \varphi^\tau
  \quad\iff\quad
  \frakB \models \varphi \quad\text{for some } \frakB \in \tau(\frakA)\,.
\end{align*}
Furthermore, if $\tau$~is quantifier-free, then the quantifier-rank
of~$\varphi^\tau$ is no larger than that of~$\varphi$.
\end{lem}
\begin{cor}
Let $\tau$~be a quantifier-free transduction and $\frakA$~and~$\frakB$
structures.
\begin{align*}
  \MTh_h(\frakA) = \MTh_h(\frakB)
  \qtextq{implies}
  \MTh_h(\tau(\frakA)) = \MTh_h(\tau(\frakB))\,.
\end{align*}
\end{cor}

\subsection{Equational classes and the Semi-Linearity Theorem}

We can use monadic second-order transductions to define
two important families of graph classes\?: the \emph{HR-equational}
and the \emph{VR-equational} classes of graphs.

The family~$\mathcal{VR}$ of \emph{VR-equational} graph classes
consists of all classes~$\calC$ such that $\lfloor\calC\rfloor$
is the image of the class~$\calT$ of all trees
under a monadic second-order transduction.
Similarly, the family~$\mathcal{HR}$ of \emph{HR-equational} graph classes
consists of all classes~$\calC$ such that $\lceil\calC\rceil$
is the image of~$\calT$ under a monadic second-order transduction.

Both families can alternatively be defined using
systems of equations in a corresponding graph algebra\?:
the VR-equational classes are the solutions of systems of equations
over the VR-algebra of graphs, i.e., the graph algebra whose operations
define clique-width, and the
HR-equational classes are the solutions of systems of equations
over the HR-algebra of graphs, i.e., the graph algebra whose operations
define tree-width.
We recall that every HR-equational class (of simple graphs) is VR-equational.

VR-equationality and HR-equationality are two possible generalisations
of the notion of a context-free language to graphs.
In light of the alternative definition in terms of systems of equations
it is not surprising that there is a close connection between
VR-equationality and clique-width and between HR-equationality and tree-width.
Every class in $\mathcal{VR}$ has bounded clique-width, while classes in
$\mathcal{HR}$ have bounded tree-width.
Conversely, every $\MSO_1$-definable class of graphs of bounded clique-width
is VR-equational and every $\MSO_2$-definable class of graphs of bounded
tree-width is HR-equational.
However, some VR-equa\-tional or HR-equational classes are not of this form.
This corresponds to the fact that some context-free languages are
not regular.

There is a third characterisation of $\mathcal{VR}$~and~$\mathcal{HR}$ in terms
of graph grammars.
VR-equational classes can be generated by \emph{vertex replacement} grammars,
while HR-equational  classes can be generated by \emph{hyperedge replacement}
grammars.
We refer the reader to the book \cite{CourcelleEngelfriet12} for details.
In the present article, we will only consider such classes specified,
as defined above, as images of trees under transductions.
Note that the definition scheme of a class~$\calC$
provides a finite representation of~$\calC$.
Consequently, we can process VR-equational and HR-equational classes
by algorithms and we can state decision problems in a meaningful way.

One important property of a VR-equational class~$\calC$ is the fact that
the spectrum of every $\MSO$-definable set predicate inside~$\calC$
is semi-linear.
Recall that a set $S \subseteq \bbN^n$ is \emph{semi-linear}
if it is a finite union of sets of the form
\begin{align*}
  P = \set{ \bar k+i_0\bar p_0 +\dots+ i_{m-1}\bar p_{m-1} }
          { i_0,\dots,i_{m-1} \in \bbN}\,,
\end{align*}
for fixed tuples $\bar k,\bar p_0,\dots,\bar p_{m-1} \in \bbN^n$.

The following result is Theorem~7.42 of \cite{CourcelleEngelfriet12}
(the fact that one can compute a representation of the semi-linear set
is not stated explicitly in \cite{CourcelleEngelfriet12},
but it follows from the proof since all of its steps are effective).
\begin{thm}[Semi-Linearity Theorem]
Let $\calC$~be a VR-equational class of graphs and let
$\varphi(X_0,\dots,X_{n-1})$ be an $\MSO$-formula.
The set
\begin{align*}
  M_\varphi(\calC) :=
  \biglset (\abs{P_0},\dots,\abs{P_{n-1}}) \bigmset {}
  & \lfloor G\rfloor \models \varphi(\bar P) \text{ for some } G = \langle V,E\rangle \in \calC \\
  & \text{and } P_0,\dots,P_{n-1} \subseteq V \bigrset
\end{align*}
is semi-linear,
and a finite representation of this set can be computed from~$\varphi$
and a representation of~$\calC$.
\end{thm}

\section{Definable orders}   
\label{Sect: general orderability}

For simplicity, we will use the term \emph{order} for linear orders.
When considering non-linear partial orders, we will explicitly speak
of \emph{partial orders.}
\begin{defi}\label{Def: orderable class}
Let $\Sigma$~be a relational signature and $\calC$~a class of $\Sigma$-structures.

(a) An $\MSO$-formula~$\varphi(x,y;\bar Z)$ \emph{defines an order}
on~$\calC$ if, for every non-empty structure $\frakA \in \calC$,
there are sets $P_0,\dots,P_{n-1} \subseteq A$ such that the formula
$\varphi(x,y;\bar P)$ defines an order on~$\frakA$.

(b) The class~$\calC$ is \emph{$\MSO$-orderable}
if there is an $\MSO$-formula~$\varphi$ defining an order on~$\calC$.

(c) A class~$\calC$ of graphs \emph{$\MSO_1$-orderable} if
the class $\lfloor\calC\rfloor$ is $\MSO$-orderable, and we call it
\emph{$\MSO_2$-orderable} if $\lceil\calC\rceil$ is $\MSO$-orderable.
\closingmark\end{defi}

\begin{rem}\label{Rem: orderability formula}
(a) For orderability by a formula $\varphi(x,y;\bar Z)$, we only require
that there are \emph{some} parameters~$\bar P$ such that $\varphi(x,y;\bar P)$
defines an order. We do not care about the behaviour of~$\varphi$ for other
values of the parameters.
We could require the formula $\varphi(x,y;\bar P')$ to be always false
for such parameters~$\bar P'$. This is no loss of generality, as we can
replace $\varphi(x,y;\bar Z)$ by the formula
\begin{align*}
  \varphi(x,y;\bar Z) \land \mathrm{ord}_\varphi(\bar Z)\,,
\end{align*}
where the formula
\begin{align*}
  \mathrm{ord}_\varphi(\bar Z)
    := {}&
      \forall x\forall y[\varphi(x,y;\bar Z) \land \varphi(y,x;\bar Z) \liff
                          x = y] \\
    {} \land {}&
       \forall x\forall y\forall z[\varphi(x,y;\bar Z) \land \varphi(y,z;\bar Z)
                                   \lso \varphi(x,z;\bar Z)]
\end{align*}
states that the relation defined by~$\varphi$ with parameters~$\bar Z$ is an
order.

(b) For every $\MSO$-formula~$\varphi(x,y;\bar Z)$ there exists
a largest class~$\calC_\varphi$ of $\Sigma$-struc\-tures that is ordered
by~$\varphi$.
This class can be defined by $\exists\bar Z\,\mathrm{ord}_\varphi(\bar Z)$.
Fixing an enumeration $\varphi_0(x,y;\bar Z),\dots,\varphi_{n-1}(x,y;\bar Z)$
of all $\MSO$-formulae of quantifier-rank~$m$ with $k$~parameters
$Z_0,\dots,Z_{k-1}$,
we obtain the class~$\calC_{m,k}$ of all $\Sigma$-structures ordered
by some of these formulae.
It~is defined by
$\exists\bar Z\Lor_{i < n} \mathrm{ord}_{\varphi_i}(\bar Z)$.
This class can be ordered by the formula
\begin{align*}
  \psi_{m,k}(x,y;\bar Z) :=
    \Lor_{i < n} \Bigl[\Land_{j < i} \neg\mathrm{ord}_{\varphi_j}(\bar Z)
                       \land \mathrm{ord}_{\varphi_i}(\bar Z)
                       \land \varphi_i(x,y;\bar Z)\Bigr]\,.
\end{align*}
It follows that any $\MSO$-orderable class~$\calC$ can be ordered by
$\psi_{m,k}$ for sufficiently large $m$~and~$k$.
\closingmark\end{rem}

\begin{rem}
By definition, a class is $\MSO_2$-orderable if,
in each graph $G = \langle V,E\rangle$,
we can define a order on the set $V \cup E$.
This is in fact equivalent to requiring just an order on the set~$V$
of vertices since, for simple graphs, any such order induces one on $V \cup E$.
For instance, we can require that every vertex is smaller than all edges, and
that an edge $(u,v)$ is smaller than an edge $(u',v')$
(orienting these pairs such that $u < v$ and $u' < v'$)
if either $u < u'$, or $u = u'$ and $v < v'$.
\closingmark\end{rem}

\begin{prop}\label{Prop: orderable classes closed under union}
Let $\calC$~and~$\calK$ be non-empty classes of $\Sigma$-structures.
\begin{enumerate}[label=\upshape(\alph*)]
\item $\calC \cup \calK$ is $\MSO$-orderable if, and only if, $\calC$ and $\calK$ are $\MSO$-orderable.
\item $\calC \oplus \calK := \set{ \frakA \oplus \frakB }{ \frakA \in \calC,\ \frakB \in \calK }$ is $\MSO$-orderable
  if, and only if, $\calC$ and $\calK$ are $\MSO$-orderable.
\end{enumerate}
\end{prop}
\begin{proof}
(a) Clearly, if $\varphi$~defines an order on $\calC \cup \calK$, it also defines orders on $\calC$
and on $\calK$.
Conversely, let $\varphi(x,y;\bar Z)$ and $\psi(x,y;\bar Z')$ be $\MSO$-formulae defining an order on,
respectively, $\calC$~and~$\calK$.
Let $\mathrm{ord}_\varphi(\bar Z)$ be the formula
(of quantifier-rank $\qr(\varphi)+3$)
from Remark~\ref{Rem: orderability formula} stating that the relation
defined by~$\varphi$ with parameters~$\bar Z$ is an order.
Then we can order $\calC \cup \calK$ by the formula
\begin{align*}
  \vartheta(x,y;\bar Z,\bar Z') :=
    [\mathrm{ord}_\varphi(\bar Z) \land \varphi(x,y;\bar Z)]
    \lor [\neg\mathrm{ord}_\varphi(\bar Z) \land \psi(x,y;\bar Z')]\,.
\end{align*}

(b) First, suppose that $\calC$~and~$\calK$ are ordered by the formulae
$\varphi(x,y;\bar Z)$ and $\psi(x,y;\bar Z')$, respectively.
We order $\calC \oplus \calK$ as follows.
Consider $\frakA \oplus \frakB \in \calC \oplus \calK$ and let
$\bar P$~and~$\bar Q$ be the parameters used by $\varphi$~and~$\psi$
to order $\frakA$~and~$\frakB$, respectively.
Using the set~$B$ as one additional parameter, we can define the order
\begin{align*}
  x \leq y \quad\defiff\quad
  &x,y \in A \text{ and } \frakA \models \varphi(x,y;\bar P) \\
  \text{or } & x,y \in B \text{ and } \frakB \models \psi(x,y;\bar Q) \\
  \text{or } & x \in A \text{ and } y \in B\,.
\end{align*}

Conversely, suppose that there is a formula $\varphi(x,y;\bar Z)$
ordering $\calC \oplus \calK$.
We construct a formula~$\psi(x,y;\bar Z)$ ordering~$\calC$.
(The orderability of~$\calK$ follows by symmetry.)
By the Composition Theorem, there exist finite lists $p_0,\dots,p_{m-1}$,
$q_0,\dots,q_{m-1}$, and $s_0,\dots,s_{n-1}$, $t_0,\dots,t_{n-1}$ of
$\MSO$-theories of quantifier-rank $h := \qr(\varphi)$ and
$h+3 = \qr(\mathrm{ord}_\varphi)$, respectively, such that, for all
$\frakA \in \calC$, $\frakB \in \calK$, $\bar P$~in $\frakA \oplus \frakB$,
and $a,b \in A$,
\begin{align*}
  \frakA \oplus \frakB \models \varphi(a,b;\bar P)
  \quad\iff\quad
  &\MTh_h(\frakA,\bar P \restriction A,a,b) = p_i \text{ and} \\
  &\MTh_h(\frakB,\bar P \restriction B) = q_i\,, \text{ for some } i < m\,, \\
\prefixtext{and}
  \frakA \oplus \frakB \models \mathrm{ord}_\varphi(\bar P)
  \quad\iff\quad
  &\MTh_{h+3}(\frakA,\bar P \restriction A) = s_i \text{ and} \\
  &\MTh_{h+3}(\frakB,\bar P \restriction B) = t_i\,, \text{ for some } i < n\,.
\end{align*}
We fix a structure $\frakB_0 \in \calK$ and set
\begin{align*}
  I := \set{ i < n }
           { \textstyle\frakB_0 \models \exists\bar Z\Land t_i(\bar Z) }\,.
\end{align*}
For each $i \in I$, we choose parameters~$\bar Q_i$ in~$\frakB_0$ such that
$\MTh_{h+3}(\frakB_0,\bar Q_i) = t_i$, and we set
\begin{align*}
  J_i := \set{ j < m }{ \MTh_h(\frakB_0,\bar Q_i) = q_j }\,.
\end{align*}
We claim that the formula
\begin{align*}
  \psi(x,y;\bar Z) :=
    \Lor_{i \in I} \Bigl[
      \Land_{\substack{k \in I\\k<i}} \neg\vartheta_k(\bar Z) \land
      \vartheta_i(\bar Z) \land \Lor_{j \in J_i} \chi_j(x,y;\bar Z))\Bigr]
\end{align*}
orders~$\calC$ where
$\vartheta_i(\bar Z) := \Land s_i$ and $\chi_i(x,y;\bar Z) := \Land p_i$.
Let $\frakA \in \calC$
and let $l \in I$ be the minimal index such that
$\frakA \models \exists\bar Z\vartheta_l(\bar Z)$.
We choose sets~$\bar P$ in~$\frakA$ such that
$\MTh_{h+3}(\frakA,\bar P) = s_l$.
By choice of $s_l$~and~$t_l$ it follows that $\varphi(x,y;\bar P \cup \bar Q_l)$
orders $\frakA \oplus \frakB_0$.
($\bar P \cup \bar Q_l$ denotes the tuple where each component is the union
of the corresponding components of $\bar P$~and~$\bar Q_l$.)
For $a,b \in A$, it further follows that
\begin{align*}
         \frakA \models \psi(a,b;\bar P)
\quad\iff\quad&\text{there is some } i \in I \text{ such that} \\
         &\MTh_{h+3}(\frakA,\bar P) = s_i\,,\\
         &\MTh_{h+3}(\frakA,\bar P) \neq s_k\,, \text{ for all } k < i\,,
          \text{ and} \\
         &\MTh_h(\frakA,\bar P,a,b) = p_j\,, \quad\text{for some } j \in J_i\,, \\
\iff\quad&\MTh_h(\frakA,\bar P,a,b) = p_j\,, \quad\text{for some } j \in J_l\,, \\
\iff\quad&\text{there is some } j < m \text{ such that} \\
         &\MTh_h(\frakA,\bar P,a,b) = p_j \qtextq{and}
          \MTh_h(\frakB_0,\bar Q_l) = q_j \\
\iff\quad&\frakA \oplus \frakB_0 \models \varphi(a,b;\bar P \cup \bar Q_l)\,.
\end{align*}
Hence, $\psi(x,y;\bar P)$ orders~$\frakA$.
\end{proof}

\begin{rem}
Every class consisting of a single (finite) structure is obviously
$\MSO$-orderable. By Proposition~\ref{Prop: orderable classes closed under union},
it follows that all finite classes are $\MSO$-orderable.
\closingmark\end{rem}

\begin{rem}\label{Rem: ordering graphs with additional edges}
Let $\calC$~be a class of graphs and let $\varphi(x,y;\bar Z)$
be an $\MSO$-formula defining an order on~$\lceil\calC\rceil$.
Let $\calC_+$~be the class of all graphs obtained from graphs in~$\calC$
by adding edges arbitrarily.
Then $\lceil\calC_+\rceil$ can be ordered by the formula
$\varphi_+(x,y;\bar Z,Z')$ obtained from~$\varphi(x,y;\bar Z)$
by replacing every atomic formula of the form $\mathrm{inc}(u,v)$
by the formula $\mathrm{inc}(u,v) \land v \in Z'$, and
by relativising every quantifier to the set~$Z'$.
(If $\bar P$~are parameters such that $\varphi(x,y;\bar P)$
orders the graph $G = \langle V,E\rangle$, then
$\varphi_+(x,y;\bar P,V \cup E)$ orders
every supergraph $G_+ = \langle V,E_+\rangle$
such that $E_+ \supseteq E$.)
\closingmark\end{rem}

\begin{rem}\label{Rem: orderable implies interpretation of paths}
Definition~\ref{Def: orderable class} can be formulated in terms of
monadic second-order transductions.
A class~$\calC$ of $\Sigma$-structures is $\MSO$-orderable if, and only if,
there exists a noncopying, domain-preserving transduction~$\sigma$ mapping
each structure $\frakA \in \calC$ to an expansion $\langle\frakA,{\leq}\rangle$
by a linear order~$\leq$.
Moreover it is easy to write down a transduction~$\tau$ mapping
any ordered structure $\langle\frakA,{\leq}\rangle$ to a path
that connects all elements of~$\frakA$.
Consequently, if $\calC$~is infinite (up to isomorphism) and $\MSO$-orderable,
we obtain an $\MSO$-transduction $\tau \circ \sigma$ mapping~$\calC$
to the class of all finite paths.
This implies that,
in the transduction hierarchy (cf.~\cite{BlumensathCourcelle10}),
the class~$\calC$ lies above the class of all paths.
\closingmark\end{rem}

The opposite of an orderable class is a class
of which no infinite subclass can be ordered.
We call such classes \emph{hereditarily unorderable.}
\begin{defi}
A class~$\calC$ of structures is \emph{hereditarily $\MSO$-unorderable,}
if it is infinite and no infinite subclass of~$\calC$ is $\MSO$-orderable.
For classes of graphs, we define the terms
\emph{hereditarily $\MSO_1$-unorderable} and
\emph{hereditarily $\MSO_2$-unorder\-able} analogously.
\closingmark\end{defi}
\begin{exa}
(a) The class $\calC = \set{ K_n }{ n \in \bbN,\ n > 0 }$ of cliques is
$\MSO_2$-orderable and hereditarily $\MSO_1$-unorderable.
To order~$K_n$, we can choose a set of edges~$P$ forming a Hamiltonian
path in~$K_n$.
Let $Q$~be a singleton set consisting of one end-point of this path.
Then we can use $P$~and~$Q$ to define a linear order on~$K_n$.

Without using $\MSO_2$-parameters, such a definition is not possible.
For each fixed number~$k$ of parameters and all sufficiently large~$n$,
every expansion of~$K_n$ by $k$~parameters $P_0,\dots,P_{k-1}$ admits
a nontrivial automorphism. Consequently, no formula can define a linear
order on $\langle K_n,\bar P\rangle$.

(b) The class~$\calT_n$ of trees of height at most~$n$ is both,
hereditarily $\MSO_1$-unorder\-able and hereditarily $\MSO_2$-unorderable.
This follows from Theorem~\ref{Thm: orderability of trees} below.
\closingmark\end{exa}

\section{$\MSO_2$-definable orderings}   
\label{Sect: MSO2}

In this section we derive characterisations for $\MSO_2$-orderable classes.
$\MSO_1$-order\-ability will be considered in Section~\ref{Sect: MSO}.

\subsection{Necessary conditions}   
\label{SSect: MSO2-necessary}

We start by providing a necessary condition for $\MSO_2$-order\-ab\-il\-ity.
Below we will then show that, for certain classes of graphs,
this condition is also sufficient.
\begin{defi}
Let $\frakA = \langle A,\bar R\rangle$ be a relational structure.

(a) We call~$\frakA$ \emph{connected} if it cannot be written as a disjoint union
$\frakA = \frakB \oplus \frakC$ of two nonempty substructures.
A \emph{connected component} of~$\frakA$ is a
maximal substructure that is connected and nonempty.

(b) For a number $k \in \bbN$, we denote by $\Sep(\frakA,k)$
the maximal number of connected components of $\frakA - S$,
where $S \subseteq A$ ranges over all sets of size at most~$k$.
For a graph~$G$, we set $\Sep(G,k) := \Sep(\lfloor G\rfloor,k)$.

(c) For a function $f : \bbN \to \bbN$, we say that a class~$\calC$ of structures has property $\SEP(f)$ if
\begin{align*}
  \Sep(\frakA,k) \leq f(k)\,,
  \quad\text{for all } \frakA \in \calC \text{ and all } k \in \bbN\,.
\end{align*}
We say that $\calC$~has property $\SEP$, if it has
property $\SEP(f)$, for some function $f : \bbN \to \bbN$.
\closingmark\end{defi}

\begin{exa}\label{Exam: Sep for bipartite graph}
For complete bipartite graphs~$K_{n,m}$ with $n \leq m$ we have
\begin{align*}
  \Sep(K_{n,m},k)
  = \begin{cases}
      1 &\text{if } k < n\,, \\
      m &\text{if } k \geq n\,.
    \end{cases}
\end{align*}
For complete $d$-partite graphs $K_{m_0,\dots,m_{d-1}}$ with $m_0 \geq\dots\geq m_{d-1}$
and $d \geq 2$, we have
\begin{align*}
  \Sep(K_{m_0,\dots,m_{d-1}},k)
  = \begin{cases}
      1   &\text{if } k < m_1+\dots+m_{d-1}\,, \\
      m_0 &\text{if } k \geq m_1+\dots+m_{d-1}\,.
    \end{cases}
\end{align*}
We leave the straightforward verification to the reader.
\closingmark\end{exa}

\begin{exa}
Let $f : \bbN \to \bbN \setminus \{0\}$ be a function and let $n \in \bbN$.
We construct a graph $G_n(f)$ such that
\begin{align*}
  \Sep(G_n(f),k) \geq f(k)\,,
  \quad\text{for all } k \leq n\,.
\end{align*}
Let $T$~be the tree of height~$n$, where every vertex~$v$
on level~$k$ has $f(k)$~immediate successors. That is,
\begin{align*}
  T := \set{ w \in \bbN^{\leq n} }{ w(k) < f(k) \text{ for all } k }\,.
\end{align*}
The desired graph $G_n(f)$ is obtained from this tree by adding all edges
$(x,y)$ such that $x \prec y$.
For a given $k \leq n$, choose a path $v_0,\dots,v_{k-1}$ of length~$k-1$
from the root~$v_0$ to some vertex $v_{k-1}$ on level $k-1$.
Removing the set $S := \{v_0,\dots,v_{k-1}\}$
we obtain a graph $G_n(f) - S$ with more than $f(k)$ connected components,
since each of the $f(k)$ immediate successors of $v_{k-1}$
belongs to a different connected component.
\closingmark\end{exa}

Let us show that having property $\SEP$ is a necessary condition for a class
to be $\MSO_2$-orderable.
\begin{prop}\label{Prop: Sep boundend}
There exists a function $f : \bbN^3 \to \bbN$ such that
$\Sep(G,k) \leq f(n,m,k)$ for every graph~$G$ such that
$\lceil G\rceil$ can be ordered by an $\MSO$-formula of the form
$\varphi(x,y;\bar P)$ where
$\qr(\varphi) \leq m$ and $\bar P = \langle P_0,\dots, P_{n-1}\rangle$
are parameters.
Furthermore, the function $f(n,m,k)$~is effectively elementary in the argument~$k$, that is,
there exists a computable function~$g$ such that $f(n,m,k) \leq \exp_{g(n,m)}(k)$.
\end{prop}
\begin{proof}
Fixing $k,m,n \in \bbN$, we define $f(n,m,k) := d$ where
$d$~is an upper bound on the number of $\MSO$-theories of the form
$\MTh_m(\lceil H\rceil,P_0,\dots,P_{n-1},v_0,\dots,v_k)$
where $H$~is a graph, $P_0,\dots,P_{n-1}$ are parameters,
and $v_0,\dots,v_k$ are vertices of~$H$.
For fixed $n$~and~$m$, we can choose $d$~to be elementary in~$k$.

Let $\varphi(x,y;\bar Z)$ be an $\MSO$-formula of quantifier-rank at most~$m$,
let $G$~be a graph with $\Sep(G,k) > f(n,m,k)$,
and let $P_0,\dots,P_{n-1}$ parameters from~$G$.
We have to show that $\varphi(x,y;\bar P)$ does not order~$\lceil G\rceil$.
Fix a set $S = \{s_0,\dots,s_{k-1}\}$ of vertices
such that $G-S$ has more than~$d$ connected components.
Fix distinct connected components $C_0,\dots,C_d$ of $G - S$
and vertices $a_i \in C_i$.
By choice of~$d$, there are indices $i < j$ such that
\begin{align*}
     & \MTh_m\bigl(\lceil G[C_i \cup S]\rceil,\bar P \restriction {(C_i \cup S)},s_0,\dots,s_{k-1},a_i\bigr) \\
{}={}& \MTh_m\bigl(\lceil G[C_j \cup S]\rceil,\bar P \restriction {(C_j \cup S)},s_0,\dots,s_{k-1},a_j\bigr)\,.
\end{align*}
As the structure
$\bigl\langle \lceil G\rceil,\bar P,s_0,\dots,s_{k-1},a_i,a_j\bigr\rangle$
is the amalgamation of the structures
\begin{align*}
  &\bigl\langle\lceil G[C_i \cup S]\rceil,\bar P \restriction {(C_i \cup S)},s_0,\dots,s_{k-1},a_i\bigr\rangle\,, \\
  &\bigl\langle\lceil G[C_j \cup S]\rceil,\bar P \restriction {(C_j \cup S)},s_0,\dots,s_{k-1},a_j\bigr\rangle\,, \\
\prefixtext{and}
  &\bigl\langle\lceil G[C_l \cup S]\rceil,\bar P \restriction {(C_l \cup S)},s_0,\dots,s_{k-1}\bigr\rangle\,, \qquad\text{for } l \neq i,j\,,
\end{align*}
over the tuple $\langle s_0,\dots,s_{k-1}\rangle$,
it therefore follows by Theorem~\ref{Thm: composition} that
\begin{align*}
  \MTh_m\bigl( \lceil G\rceil,\bar P,s_0,\dots,s_{k-1},a_i,a_j\bigr)
  = \MTh_m\bigl( \lceil G\rceil,\bar P,s_0,\dots,s_{k-1},a_j,a_i\bigr)\,.
\end{align*}
In particular,
\begin{align*}
  G \models \varphi(a_i,a_j;\bar P)
  \quad\iff\quad
  G \models \varphi(a_j,a_i;\bar P)\,.
\end{align*}
Hence, $\varphi(x,y;\bar P)$ does not define an order.
\end{proof}

\begin{cor}\label{Cor: GSO-orderable implies SEP}
An $\MSO_2$-orderable class of graphs~$\calC$
has property $\SEP(f)$, for an elementary function~$f$.
\end{cor}

The converse does not hold.
For instance, according to Theorem~\ref{Thm: GSO-orderable d-partite graphs}
below, the class of bipartite graphs of the form $K_{n,2^{2^n}}$ is not
$\MSO_2$-orderable,
while we have seen in Example~\ref{Exam: Sep for bipartite graph} that
it has property $\SEP(f)$ for the elementary function~$f$ such that
$f(n) = 2^{2^n}$.
Our objective therefore is to get converse results for particular classes
of graphs satisfying certain combinatorial conditions.

\begin{rem}
We have noted in Remark~\ref{Rem: ordering graphs with additional edges} that,
if a graph~$G$ can be ordered by an $\MSO_2$-formula~$\varphi$,
we can construct from~$\varphi$ a $\MSO_2$-formula~$\psi$
ordering every graph~$H$ obtained from~$G$ by adding edges.
In this case, we further have $\Sep(H,k) \leq \Sep(G,k)$, for all~$k$.
\closingmark\end{rem}

\begin{rem}
All results of Section~\ref{Sect: MSO2} also hold for directed graphs since
there is an $\MSO_2$-formula with two parameters that defines an orientation
of every undirected graph (see Proposition~9.46 of~\cite{CourcelleEngelfriet12}).
It follows that a class of directed graphs is $\MSO_2$-orderable if, and only if,
the corresponding class of undirected graphs is.
This is different for $\MSO_1$-orderability.
\closingmark\end{rem}

As a simple introductory example, let us consider classes of trees.
\begin{thm}\label{Thm: orderability of trees}
Let $\calT$~be a class of (undirected) trees.
The following statements are equivalent\?:
\begin{enumerate}
\item $\calT$~is $\MSO_1$-orderable.
\item $\calT$~is $\MSO_2$-orderable.
\item $\calT$~has property $\SEP$.
\item There exists a number $d \in \bbN$ such that every tree in~$\calT$ has maximal degree at most~$d$.
\end{enumerate}
\end{thm}
\begin{proof}
(1)~$\Rightarrow$~(2) is trivial.

(2)~$\Rightarrow$~(3) has been shown in Corollary~\ref{Cor: GSO-orderable implies SEP}.

(3)~$\Rightarrow$~(4) Suppose that $\calT$~has property $\SEP(f)$ and let $T \in \calT$.
Every vertex $v \in T$ has at most $f(1)$ neighbours since $T - \{v\}$ has at most $f(1)$ connected components.
Consequently, the maximal degree of~$T$ is bounded by~$f(1)$.

(4)~$\Rightarrow$~(1)
Let $T$~be a tree with maximal degree at most~$d$.
We use $d$~parameters $P_0,\dots,P_{d-1}$ to order~$T$.
Fixing a vertex $r \in T$ as root, we obtain an injective
embedding $g : T \to d^{<m}$, for some number $m \in \bbN$.
We set
\begin{align*}
  P_i := \set{ v \in T }{ g(v) = wi \text{ for some } w }\,.
\end{align*}
Note that $r$~is the only vertex of~$T$ that is not contained in any of these sets.
Hence, using~$\bar P$, we can define the tree-order~$\preceq$ on~$T$.
We can also define the lexicographic ordering\?:
\begin{align*}
  u \leq v \quad\defiff\quad
  u \preceq v\,, \text{ or } &u_0 \in P_i,\ v_0 \in P_k, \text{ for } i < k,
  \text{ where } u_0, v_0 \text{ are the}\\
  &\text{immediate successors of the longest common} \\
  &\text{prefix of } u \text{ and } v
   \text{ with } u_0 \preceq u \text{ and } v_0 \preceq v\,.
\end{align*}
\upqed
\end{proof}
\begin{cor}\label{Cor: trees of bounded depth hereditarily MSO2-unorderable}
Let $k \in \bbN$. The class of trees of height at most~$k$ is hereditarily
$\MSO_2$-unorderable.
\end{cor}
\begin{proof}
For any given height~$k$, there are only finitely many trees
(up to isomorphism) satisfying condition~(4) of the theorem.
\end{proof}

\subsection{Omitting a minor}   
\label{Sect: ommiting a minor}

We start by presenting a characterisation for classes of graphs omitting a
fixed graph as minor (for an introduction to graph minors see,
e.g.,~\cite{Diestel10}). For short, we will say that such a class
\emph{omits a minor.}
Recall that a spanning forest~$F$ of a graph~$G$ is defined to be directed.
A spanning forest~$F$ is \emph{normal} if the ends of every edge of~$G$
are comparable with respect to the tree-order~$\preceq_F$ on~$F$
(see, e.g., Section~1.5 of~\cite{Diestel10}).
\begin{defi}
Let $G$~be a graph and $F \subseteq G$ a normal spanning forest of~$G$.

(a) We denote the set of predecessors of a vertex~$x$ by
\begin{align*}
  \Pred_F(x) := \set{ y }{ y \prec_F x }\,.
\end{align*}

(b) For $x \in G$, we define
\begin{align*}
  B_F(x) := \set{ v \prec_F x }
                { \text{there is an edge } (u,v) \text{ of } G
                  \text{ such that } x \preceq_F u }\,.
\end{align*}
\upqed
\closingmark\end{defi}

\begin{lem}\label{Lem: simple conditions on separators}
Let $G$~be a graph, $F$~a normal spanning forest of~$G$, $x \in G$,
and $B \subseteq \Pred_F(x)$.
\begin{enumerate}[label=\upshape(\alph*)]
\item If $\abs{B} \geq p$ and there are $p$~immediate successors~$y$ of~$x$
  such that $B_F(y) = B \cup \{x\}$, then $K_{p,p}$ is a minor of~$G$.
\item If $\abs{B} < p$ and $\Sep(G,p) \leq d$,
  then there are at most~$d$ immediate successors~$y$ of~$x$
  such that $B_F(y) = B \cup \{x\}$\,.
\end{enumerate}
\end{lem}
\begin{proof}
(a)
Suppose that there are $p$~distinct immediate successors $y_0,\dots,y_{p-1}$ of~$x$ with $B(y_i) = B \cup \{x\}$
and fix distinct vertices $b_0,\dots,b_{p-1} \in B$.
Let $H$~be the minor of~$G$ obtained by contracting the subtrees rooted at
$y_0,\dots,y_{p-1}$ to single vertices $\widetilde y_0,\dots,\widetilde y_{p-1}$
and by removing all remaining vertices except for
$\widetilde y_0,\dots,\widetilde y_{p-1}$ and $b_0,\dots,b_{p-1}$.
Then $H \cong K_{p,p}$.

(b) Set $S := B \cup \{x\}$ and let $y_0,\dots,y_{n-1}$ be an enumeration of
all immediate successors of~$x$ such that $B(y_i) = S$.
Then $y_0,\dots,y_{n-1}$ lie in different connected components of $G - S$.
Hence, $n \leq \Sep(G,p) \leq d$.
\end{proof}

\begin{thm}\label{Thm: ordering graphs without Kpp minor}
For every $p,d \in \bbN$, the class $\calC_{p,d}$ of all graphs~$G$
that satisfy $\Sep(G,p) \leq d$ and that do not contain $K_{p,p}$ as a minor
is $\MSO_2$-orderable.
\end{thm}
\begin{proof}
Consider a graph $G \in \calC_{p,d}$.
Let $F$~be a normal spanning forest of~$G$.
Since $G$~has $\Sep(G,0) \leq d$ connected components,
the forest~$F$ has at most~$d$ roots.
Recall that a forest is oriented with edges pointing away from the roots.
We can encode~$F$ by two parameters\?:
its set of edges and its set of roots.
(Since the first set consists of edges and the second one of vertices,
we could even take their union as a single parameter.)
We will use a lexicographic order on~$F$ to order~$G$,
based on orderings (i)~of the roots of~$F$ and
(ii)~of the immediate successors of every vertex of~$F$.

Consider a vertex $x \in F$ with immediate successors $y_0,\dots,y_{m-1}$.
Since each set $B_F(y_i)$ is linearly ordered by~$\preceq_F$,
we can define a preorder on the immediate successors
by using the lexicographic ordering of the sets $B_F(y_i)$\?:
\begin{align*}
  y_i \sqsubseteq y_k \quad\defiff\quad B_F(y_i) \leq_\lex B_F(y_k)\,.
\end{align*}
To prove that there is a definable order extending this preorder,
it is sufficient to show that the equivalence classes of this preorder
have bounded cardinality.
Let $k := \max {\{p,d\}}$. For every set $B \subseteq \Pred_F(x)$,
there are at most~$k$ immediate successors~$y_i$ of~$x$ with
$B_F(y_i) = B \cup \{x\}$\?:
for $\abs{B} \geq p$, this follows from
Lemma~\ref{Lem: simple conditions on separators}\,(a)\?;
for $\abs{B} < p$, it follows from
Lemma~\ref{Lem: simple conditions on separators}\,(b).

The parameters needed to define the desired linear order consist of
the set of edges of the spanning forest~$F$
and $d+k$ parameters to distinguish and order the roots of~$F$ and to order
the immediate successors~$y$ of a vertex~$x$ that have the same set~$B_F(y)$.
\end{proof}

\begin{thm}\label{Thm: GSO-orderable with excluded minor}
Let $\calC$~be a class of graphs omitting a minor~$H$.
The following statements are equivalent\?:
\begin{enumerate}
\item $\calC$~is $\MSO_2$-orderable.
\item $\calC$~has property $\SEP$.
\item $\calC$~has property $\SEP(f)$ for some elementary function~$f$.
\end{enumerate}
Furthermore, given~$H$ we can compute a number~$k$ such that we can
replace $\SEP(f)$ by $\SEP(\exp_k)$ in~\textup{(3)}.
\end{thm}
\begin{proof}
$(1) \Rightarrow (3)$ follows by Corollary~\ref{Cor: GSO-orderable implies SEP}
and $(3) \Rightarrow (2)$ is trivial.

For $(2) \Rightarrow (1)$, suppose that $\calC$~has property $\SEP(f)$.
By Theorem~\ref{Thm: ordering graphs without Kpp minor},
all classes $\calC_{p,d}$ are $\MSO_2$-orderable.
Since every graph with $n$~vertices and $m$~edges is a minor of~$K_{n,m}$,
we can choose $p$ sufficiently large such that $H$~is a minor of~$K_{p,p}$.
Set $d := f(p)$.
Then $\calC \subseteq \calC_{p,d}$ and it follows that
$\calC$~is also $\MSO_2$-orderable.
\end{proof}
\begin{rem}\label{Rem: bounded tree-width is MSO2-orderable}
(a)
For each $k \in \bbN$, the class of graphs of tree-width at most~$k$
excludes some (planar) graph as a minor and, hence,
it satisfies the conditions of Theorem~\ref{Thm: GSO-orderable with excluded minor}.

(b)
Although this fact is not directly related to our work,
we mention that Grohe has proved
that every class of graphs excluding a minor
is orderable in least fixed-point logic.
It follows that least fixed-point logic captures PTIME on these
classes \cite{Grohe10,GroheXX}.
\closingmark\end{rem}
In contrast to Remark~\ref{Rem: bounded tree-width is MSO2-orderable}\,(a),
we have the following result for classes of graphs of bounded
\emph{$n$-depth tree-width} (which is defined as tree-width, but where we
only consider tree decompositions with index trees of height at most~$n$).
This graph complexity measure was introduced in~\cite{BlumensathCourcelle10}.
\begin{prop}\label{Prop: n-depth tree-width hereditarily MSO2-unorderable}
Let $n,k \in \bbN$.
A class of graphs of $n$-depth tree-width at most~$k$
is $\MSO_2$-orderable if, and only if, it is finite.
Hence, the class of all graphs of $n$-depth tree-width at most~$k$ is
hereditarily $\MSO_2$-unorderable.
\end{prop}
\begin{proof}
Let $\calC$~be an infinite class of graphs of $n$-depth tree-width at most~$k$.
As we have argued in
Remark~\ref{Rem: orderable implies interpretation of paths},
if $\calC$~were $\MSO_2$-orderable, we could define an
$\MSO_2$-transduction mapping it to the class of all finite paths.
This is not possible by Theorem~6.4 of~\cite{BlumensathCourcelle10}.
\end{proof}

In the following we try to compute a better bound on the function~$f$ in
Theorem~\ref{Thm: GSO-orderable with excluded minor}\,(3).
We can improve the bound from elementary to singly exponential.
\begin{lem}\label{Lem: number of components containing a successor}
Let $G$~be a graph such that $\Sep(G,p) \leq d$ and $K_{p,p}$
is not a minor of~$G$.
Let $F$~be a normal spanning forest of~$G$ and $S$~a set of at most~$k$ vertices of~$G$.
For every vertex $x \in S$, at most $k + 2^k \cdot \max{\{p,d\}}$
connected components of $G - S$ contain an immediate successor of~$x$ (in~$F$).
\end{lem}
\begin{proof}
Let $s_0 \prec_F \dots \prec_F s_{m-1} = x$ be an enumeration of
$\Pred_F(x) \cup \{x\}$.
For an immediate successor~$y$ of~$x$, we define
\begin{align*}
  I(y) := \set{ i < m }
              { \text{there is some } z \in B_F(y) \text{ such that }
                z \prec_F s_i \text{ and }
                (i = 0 \text{ or } s_{i-1} \prec_F z) }\,.
\end{align*}
If $y$~and~$y'$ are immediate successors of~$x$ in different connected components of $G - S$, then
$I(y) \cap I(y') = \emptyset$. Consequently, there are at most $m \leq k$
connected components of $G - S$ containing an immediate successor~$y$ of~$x$
such that $I(y) \neq \emptyset$.

It remains to show that there are at most $2^k\cdot\max {\{p,d\}}$ components
of $G - S$ containing an immediate successor~$y$ with $I(y) = \emptyset$.
Every such immediate successor~$y$ satisfies $B(y) \subseteq S$.
Hence, $B(y)$~can take at most $2^m \leq 2^k$ values and,
according to Lemma~\ref{Lem: simple conditions on separators},
for each such value $B \subseteq S$ there are
at most $\max {\{p,d\}}$ immediate successors~$y$ with $B(y) = B$.
\end{proof}

\begin{prop}
Let $G$~be a graph such that $\Sep(G,p) \leq d$ and $K_{p,p}$ is not a minor
of~$G$. Then
\begin{align*}
  \Sep(G,k) \leq d + k^2 + k2^k \cdot \max {\{p,d\}}\,,
  \quad\text{for }  k \geq p\,.
\end{align*}
\end{prop}
\begin{proof}
Let $F$~be a normal spanning forest of~$G$ and $S$~a set of at most~$k$ vertices of~$G$.
We have seen in Lemma~\ref{Lem: number of components containing a successor}
that, for every vertex $x \in S$, at most $k + 2^k \cdot \max{\{p,d\}}$
connected components of $G - S$ contain an immediate successor of~$x$.
Since every connected component of $G - S$ contains
a root of~$F$ or the immediate successor of some $x \in S$,
there are at most $d + k(k + 2^k \cdot \max {\{p,d\}})$ such components.
\end{proof}

Every class omitting some minor~$H$ also omits $K_{p,p}$ as a minor,
for all sufficiently large~$p$.
The following corollary states that, in order to determine whether
such a class is $\MSO_2$-orderable, it is sufficient to bound the numbers
$\Sep(G,p)$ as opposed to the function $k \mapsto \Sep(G,k)$.
\begin{cor}\label{Cor: orderable if Sep is bounded}
Let $p \in \bbN$.
A class~$\calC$ of graphs omitting $K_{p,p}$ as a minor
is $\MSO_2$-orderable if, and only if,
\begin{align*}
  \sup {\set{ \Sep(G,p) }{ G \in \calC }} < \infty\,.
\end{align*}
\end{cor}

\begin{rem}
Graphs omitting a minor~$H$ are $r$-sparse (cf.~Definition~\ref{Def: sparse}),
for some number~$r$ depending on~$H$.
Since, for $r$-sparse graphs, the expressive powers of $\MSO_1$ and $\MSO_2$
coincide, it follows that the criterion in
Corollary~\ref{Cor: orderable if Sep is bounded} also characterises
$\MSO_1$-orderability.
\closingmark\end{rem}

\begin{rem}
The proof technique of Theorem~\ref{Thm: ordering graphs without Kpp minor}
can be extended to order certain classes of graphs that do not omit any graph
as a minor. We give two examples.

(a)
First, let us consider the class of graphs~$H_p$, for $p \geq 1$,
defined as follows.
The set of vertices of~$H_p$ is
\begin{align*}
  V := \{*\} \cup [p] \cup [p] \times S_p\,,
\end{align*}
where $S_p$~is the set of permutations of~$[p]$.
The graph~$H_p$ has the following edges\?:
\begin{alignat*}{-1}
  &(*,0) && \\
  &(*,(0,\sigma)) &&\quad\text{for } \sigma \in S_p\,, \\
  &(i,i+1)  &&\quad\text{for } i \in [p],\ i < p-1\,, \\
  &((i,\sigma),(i+1,\sigma)) &&\quad\text{for } i \in [p],\ \sigma \in S_p,\ i < p-1\,, \\
  &(i,(\sigma(i),\sigma)) &&\quad\text{for } i \in [p],\ \sigma \in S_p,\ i < p\,.
\end{alignat*}
The graph~$H_2$ is shown in Figure~\ref{Fig:H2}.
($e$~is the identity and $\tau$~is the transposition of $0$~and~$1$.)
\begin{figure}\centering
%
%
%
%
%
%
%
\includegraphics{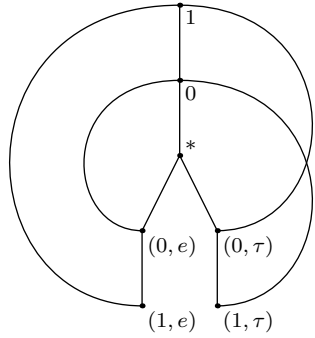}
\caption{The graph $H_2$.\label{Fig:H2}}
\end{figure}
Note that the vertex~$*$ has degree $1 + p!$.
Clearly, $H_p$~contains $K_{p,p!}$ as a minor.
Nevertheless, the class of graphs~$H_p$ is $\MSO_2$-orderable.
We can use a spanning tree whose root is the vertex $p-1$
and whose edges consist of the first four of the above types.
To compare two immediate successors $(0,\sigma)$ and $(0,\tau)$
of the vertex~$*$, we can use a lexicographic order on~$S_p$
(where we identify a permutation~$\sigma$ with the sequence
$\sigma(0)\dots\sigma(p-1)$).
Since each $H_p$~is $2$-sparse (as~it has an orientation of indegree~$2$,
cf.~Proposition~9.40 of~\cite{CourcelleEngelfriet12}),
it follows that the class is even $\MSO_1$-orderable
(cf.~Theorem~9.37 of~\cite{CourcelleEngelfriet12}).

(b) Another example is the class of cliques.
It is $\MSO_2$-orderable and does not omit a minor.
If we replace each edge by a path of length~$2$, we obtain a class of
$2$-sparse graphs that is $\MSO_2$-orderable and that still does not omit a
minor.
\closingmark\end{rem}

\begin{rem}
It is not possible to extend
Theorem~\ref{Thm: GSO-orderable with excluded minor} to $r$-sparse graphs.
A counterexample is given by the class~$\calC$ of all graphs
obtained from a bipartite graph of the form $K_{n,f(n)}$ by replacing every
edge by a path of length~$2$, where $f : \bbN \to \bbN$ is a fixed
non-elementary function.
This is a class of $2$-sparse graphs with property $\SEP$ that,
according to Corollary~\ref{Cor: GSO-orderable implies SEP},
is not $\MSO_2$-orderable.
\closingmark\end{rem}

\subsection{Deciding $\MSO_2$-orderability}   
\label{Sect: deciding orderability}

In Theorem~\ref{Thm: GSO-orderable with excluded minor} above,
we have presented a combinatorial property characterising
$\MSO_2$-order\-ability for classes of graphs omitting a minor.
A natural question is whether this property is decidable.
Of course, this question does only make sense for classes of graphs
that can be described in a finitary way.
Therefore, we will concentrate on
HR-equational and VR-equational classes.

\begin{prop}
It is decidable whether a VR-equational class~$\calC$ has property $\SEP$.
\end{prop}
\begin{proof}
Let $\calC$~be a VR-equational class and
let $\varphi(X,Y)$ be an $\MSO$-formula expressing, for a graph~$G$,
that the set~$Y$ contains exactly one vertex of each connected component of $G - X$.
The class~$\calC$ has property $\SEP$ if, and only if, there exists a function~$f$
such that, for all $G = \langle V,E\rangle \in \calC$ and $P,Q \subseteq V$,
\begin{align*}
  G \models \varphi(P,Q) \qtextq{implies} \abs{Q} \leq f(\abs{P})\,.
\end{align*}

According to the Semi-Linearity Theorem, the set
\begin{align*}
  M(\calC) :=
    \bigset{ (\abs{P}, \abs{Q}) }
           { G \models \varphi(P,Q) \text{ for some}
             G = \langle V,E\rangle \in \calC \text{ and }
             P,Q \subseteq V }
\end{align*}
is semi-linear and an effective description of $M(\calC)$ can be computed from a system of equations for~$\calC$.
Using this description, we can check whether or not, for every $n \in \bbN$,
the set $\set{ p }{ (n,p) \in M(\calC) }$ is bounded.
This is the case if, and only if, $\calC$ has property~$\SEP$.
\end{proof}

\begin{cor}\label{Cor: orderability decidable}
For an $\mathrm{HR}$-equational class~$\calC$,
it is decidable whether $\calC$~is $\MSO_2$-orderable.
\end{cor}
\begin{proof}
An $\mathrm{HR}$-equational class~$\calC$
has bounded tree-width (Proposition~4.7 of~\cite{CourcelleEngelfriet12}) and, hence,
omits some $K_{p,p}$ as a minor.
Since $\mathrm{HR}$-equational classes (of simple graphs) are VR-equational,
it follows from Theorem~\ref{Thm: GSO-orderable with excluded minor} that
$\calC$~is $\MSO_2$-orderable if, and only if, it has property~$\SEP$.
The latter is decidable by the above proposition.
\end{proof}

\begin{rem}
An alternative decidability proof can be based on Corollary~\ref{Cor: orderable if Sep is bounded}.
As the tree-width of $K_{p,p}$ is~$p$,
every class~$\calC$ of tree-width at most $p-1$ omits $K_{p,p}$ as a minor.
Furthermore, an upper bound on the tree-width of an $\mathrm{HR}$-equational class~$\calC$
can be computed from a system of equations for~$\calC$
(see Proposition 4.7 of~\cite{CourcelleEngelfriet12}).
By Corollary~\ref{Cor: orderable if Sep is bounded},
$\calC$~is $\MSO_2$-orderable if, and only if,
the set $\set{ \Sep(G,p) }{ G \in \calC }$ is bounded.
To check this condition, we consider the formula~$\varphi(X)$ expressing that
there exists a set~$S$ of size $\abs{S} \leq p$ such that $X$~contains
exactly one vertex of each connected component of $G - S$.
By the Semi-Linearity Theorem, we can compute a representation of the semi-linear set
\begin{align*}
  M(\calC) :=
    \bigset{ \abs{P} }{ G \models \varphi(P) \text{ for some } G = \langle V,E\rangle \in \calC \text{ and } P \subseteq V }\,.
\end{align*}
Using this representation we can check whether or not $M(\calC)$ is finite.
\closingmark\end{rem}

For VR-equational classes we do not obtain decidability since we cannot apply
Theorem~\ref{Thm: GSO-orderable with excluded minor}. We conjecture that a corresponding
statement holds also for these classes.
\begin{conj}\label{Conj: VR-equational with SEP is orderable}
Every VR-equational class that has property $\SEP$ is $\MSO_2$-orderable.
\end{conj}
Below we will prove this conjecture for the special cases of complete
$d$-partite graphs (Corollary~\ref{Cor: VR-equational d-partite graphs})
and chordal graphs
(Corollary~\ref{Cor: orderability of VR-equational chordal graphs}).

\subsection{Dense graphs}   

We have characterised $\MSO_2$-orderability
in Theorem~\ref{Thm: GSO-orderable with excluded minor} for classes excluding a minor.
The graphs in such classes are sparse.
In this section and the next one,
we consider the opposite extreme of certain dense graphs,
in particular, multi-partite graphs and chordal graphs.

\begin{lem}\label{Lem: ordering supergraphs of bipartite graphs}
Let $s,r \in \bbN$ and let $\calC$~be a class of graphs
such that each $G \in \calC$ is obtained from some $K_{n,m}$ with
$n \leq m \leq 2^{sn+r}$ by possibly adding new edges.
Then $\calC$~is $\MSO_2$-orderable.
\end{lem}
\begin{proof}
Consider a graph $G = \langle V,E\rangle \in \calC$ obtained by adding new edges
from a bipartite graph $K_{n,m}$ where $n \leq m \leq 2^{sn+r}$
(see also Remark~\ref{Rem: ordering graphs with additional edges}).
If $n = 0$, then $G$~has $m \leq 2^r$ vertices and we can order~$G$
using $r$~parameters.
Thus, it remains to consider the case where $n > 0$.
Since $m \leq 2^{sn+r} \leq 2^{(s+r)n}$, there exists an injective function
$\mu : [m] \to \PSet([(s+r)n])$.
Fixing enumerations $a_0,\dots,a_{n-1}$ and $b_0,\dots,b_{m-1}$ of the two
vertex classes of~$K_{n,m}$,
we define an ordering of~$G$ using the following parameters.
\begin{alignat*}{-1}
  A &:= \set{ a_i }{ i < n } \subseteq V\,, \\
  B &:= \set{ b_i }{ i < m } \subseteq V\,, \\
  S &:= \set{ (a_i,b_j) }{ i \leq j } \subseteq E\,, \\
  R_k &:= \set{ (a_i,b_j) }{ kn+i \in \mu(j) } \subseteq E\,,
   \quad\text{for } k < s+r\,.
\end{alignat*}
First, we define a strict order~$<_A$ on~$A$ by
\begin{align*}
  u <_A v \quad\defiff\quad
  u \neq v
  \text{ and, for all } x \in B, (u,x) \in S \Rightarrow (v,x) \in S\,.
\end{align*}
By definition of~$S$, this order is linear.
We extend it to all vertices of~$G$ by defining
$u < v$ if, and only if, one of the following conditions holds\?:
\begin{itemize}
\item $u,v \in A$ and $u <_A v$.
\item $u \in A$ and $v \in B$.
\item $u,v \in B$, $u \neq v$, and, if $k$~is the minimal number such that,
  for some $x \in A$,
  \begin{align*}
    (x,u) \in R_k \Leftrightarrow (x,v) \notin R_k,
  \end{align*}
  and if $x \in A$ is the $<_A$-least element with this property, then
  $(x,u) \in R_k$ and $(x,v) \notin R_k$.
  \qedhere
\end{itemize}
\end{proof}

\noindent The technique employed in this proof will be used several times in this article.
Given an already defined order on a set~$A$, we can order the vertices
not in~$A$ using the lexicographic ordering on their sets of
neighbours in~$A$.

\begin{lem}
A class~$\calC$ of complete bipartite graphs is $\MSO_2$-orderable if,
and only if, there exists a constant~$s$ such that
\begin{align*}
  K_{n,m} \in \calC \text{ with } n \leq m \qtextq{implies}
  m \leq 2^{s(n+1)}.
\end{align*}
\end{lem}
\begin{proof}
$(\Leftarrow)$
is a special case of Lemma~\ref{Lem: ordering supergraphs of bipartite graphs}.

$(\Rightarrow)$
Suppose that $\calC$~is ordered by an $\MSO$-formula $\varphi(x,y;\bar Z)$
with $s$~set variables $Z_0,\dots,Z_{s-1}$.
We claim that there is no $K_{n,m} \in \calC$ such that $m > 2^{s(n+1)}$.

For a contradiction, suppose that there is such a graph $K_{n,m} \in \calC$.
Let $\bar P$ be the parameters such that $\varphi(x,y;\bar P)$ orders $\lceil K_{n,m}\rceil$.
We enumerate the two vertex sets of~$K_{n,m}$ as $a_0,\dots,a_{n-1}$ and
$b_0,\dots,b_{m-1}$.
Since $m > 2^{s(n+1)}$ there is a subset $I \subseteq [m]$
of cardinality $\abs{I} > 2^{s(n+1)}/2^s = 2^{sn}$ such that
\begin{align*}
  b_i \in P_l \Leftrightarrow b_j \in P_l
  \quad\text{for all } i,j \in I \text{ and all } l < s\,.
\end{align*}
Similarly, there is a subset $J \subseteq I$ of cardinality
$\abs{J} > 2^{sn}/2^{sn} = 1$ such that
\begin{align*}
  (a_k,b_i) \in P_l \Leftrightarrow (a_k,b_j) \in P_l
  \quad\text{for all } i,j \in J \text{ and all } l < s \text{ and } k < n\,.
\end{align*}
Hence, there are at least two indices $i < j$ in~$J$.
The mapping $\pi : K_{n,m} \to K_{n,m}$ that interchanges $b_i$~and~$b_j$
and leaves every other vertex fixed
is an automorphism of the structure $\langle \lceil K_{n,m}\rceil,\bar P\rangle$.
Hence,
\begin{align*}
  \lceil K_{n,m}\rceil \models \varphi(b_i,b_j;\bar P)
  \quad\iff\quad
  \lceil K_{n,m}\rceil \models \varphi(b_j,b_i;\bar P)\,,
\end{align*}
and $\varphi$~does not define an order on $K_{n,m}$.
A~contradiction.
\end{proof}

\begin{lem}\label{Lem: ordering multipartite graphs}
Let $\calC$~be a class of graphs of the form $K_{m_0,\dots,m_{d-1}}$ where
\begin{align*}
  d > 2 \qtextq{and}
  m_1+\dots+m_{d-1} \geq m_0 \geq m_1 \geq\dots\geq m_{d-1} \geq 1\,.
\end{align*}
Then $\calC$~is $\MSO_2$-orderable.
\end{lem}
\begin{proof}
Consider $K_{m_0,\dots,m_{d-1}} \in \calC$ with
$m_0 \geq\dots\geq m_{d-1} \geq 1$.
Let $A_0,\dots,A_{d-1}$ be the vertex sets of this graph and let
$a^k_0,\dots,a^k_{m_k-1}$ be an enumeration of~$A_k$.
Using the parameter
\begin{align*}
  R := \set{ (a^k_0,a^{k+1}_0) }{ 0 \leq k < d-1 }
\end{align*}
we can define the preorder
\begin{align*}
  u \sqsubseteq v \quad\defiff\quad
  u \in A_i \text{ and } v \in A_k \text{ for all } i \leq k\,.
\end{align*}
As usual, we write
\begin{align*}
  u \equiv v    &\quad\defiff\quad u \sqsubseteq v \text{ and } v \sqsubseteq u\,, \\
  u \sqsubset v &\quad\defiff\quad u \sqsubseteq v \text{ and } v \nsqsubseteq u\,.
\end{align*}
Using the parameter
$S := \set{ (a^k_i,a^{k+1}_j) }{ i \leq j }$
and~$\sqsubseteq$, we can define a linear order~$\leq_B$ on $B := A_1 \cup\dots\cup A_{d-1}$
by setting $u \leq_B v$ if, and only if,
\begin{itemize}
\item $u \sqsubset v$ or
\item $u \equiv v$ and, for all~$x \sqsubset u$, $(x,u) \in S$ implies $(x,v) \in S$.
\end{itemize}
Hence, it remains to define a linear order~$\leq_A$ on~$A_0$.
Since $m_0 \leq m_1+\dots+m_{d-1}$, we can fix an enumeration $b_0,\dots,b_{n-1}$ of~$B$
and use the parameter
$S_0 := \set{ (a^0_i,b_j) }{ i \leq j }$
to define such an order.
\end{proof}

\begin{thm}\label{Thm: GSO-orderable d-partite graphs}
Let $\calC$~be a class of graphs that are all complete $d$-partite
for some $d \in \bbN$.
(We do not require the number~$d$ to be the same for every graph.)
The following statements are equivalent\?:
\begin{enumerate}
\item $\calC$~is $\MSO_2$-orderable.
\item There exists a constant~$s$ such that $\calC$~has property $\SEP(f)$ 
  where $f(k) = 2^{s(k+1)}$.
\item There exists a constant~$s$ such that
  \begin{align*}
    K_{m_0,\dots,m_{d-1}} \in \calC
    \qtextq{implies}
    M \leq 2^{s(N - M + 1)}
  \end{align*}
  where $M := \max_{i < d} m_i$ and $N := \sum_{i < d} m_i$.
\end{enumerate}
\end{thm}
\begin{proof}
$(3) \Rightarrow (1)$
Consider $K_{m_0,\dots,m_{d-1}} \in \calC$ with $m_0 \geq\dots\geq m_{d-1} \geq 1$.
We distinguish several cases.
\begin{itemize}
\item If $d \leq 2$, the claim follows by Lemma~\ref{Lem: ordering supergraphs of bipartite graphs}.
\item If $d > 2$ and $M \geq N-M$, we have
  $K_{N-M,M} \subseteq K_{m_0,\dots,m_{d-1}}$ and the claim follows
  by Remark~\ref{Rem: ordering graphs with additional edges} and
  Lemma~\ref{Lem: ordering supergraphs of bipartite graphs}.
\item If $d > 2$ and $M < N-M$ the claim follows by Lemma~\ref{Lem: ordering multipartite graphs}.
\end{itemize}

$(1) \Rightarrow (3)$
Suppose that $\lceil\calC\rceil$~is ordered by an $\MSO$-formula $\varphi(x,y;\bar Z)$
with $s$~set variables $Z_0,\dots,Z_{s-1}$.
We claim that there is no $K_{m_0,\dots,m_{d-1}} \in \calC$ with $M > 2^{s(N-M)+s}$.

For a contradiction, suppose that there is such a graph $K_{m_0,\dots,m_{d-1}} \in \calC$.
Let $\bar P$~be parameters such that $\varphi(x,y;\bar P)$ orders $K_{m_0,\dots,m_{d-1}}$.
Let $A$~be a vertex set of~$K_{m_0,\dots,m_{d-1}}$ of size~$M$ and
let $B$~be its complement.
We enumerate $A$~and~$B$ as $a_0,\dots,a_{M-1}$ and $b_0,\dots,b_{N-M-1}$, respectively.
Since $M > 2^{s(N-M)+s}$ there is a subset $I \subseteq [M]$
of cardinality $\abs{I} > 2^{s(N-M)+s}/2^s = 2^{s(N-M)}$ such that
\begin{align*}
  a_i \in P_l \Leftrightarrow a_j \in P_l
  \quad\text{for all } i,j \in I \text{ and all } l < s\,.
\end{align*}
Similarly, there is a subset $J \subseteq I$ of cardinality
$\abs{J} > 2^{s(N-M)}/2^{s(N-M)} = 1$
such that
\begin{align*}
  (a_i,b_k) \in P_l \Leftrightarrow (a_j,b_k) \in P_l
  \quad\text{for all } i,j \in J\,,\ l < s\,, \text{ and } k < N-M\,.
\end{align*}
Hence, there are at least two different indices $i,j \in J$.
The mapping $\pi : K_{m_0,\dots,m_{d-1}} \to K_{m_0,\dots,m_{d-1}}$
that interchanges $a_i$~and~$a_j$ and leaves every other vertex fixed
is an automorphism of the structure $\langle \lceil K_{m_0,\dots,m_{d-1}}\rceil,\bar P\rangle$.
Hence,
\begin{align*}
  \lceil K_{m_0,\dots,m_{d-1}}\rceil \models \varphi(a_i,a_j;\bar P)
  \quad\iff\quad
  \lceil K_{m_0,\dots,m_{d-1}}\rceil \models \varphi(a_j,a_i;\bar P)\,,
\end{align*}
and $\varphi$~does not define an order on $K_{m_0,\dots,m_{d-1}}$.
A~contradiction.

$(3) \Rightarrow (2)$
Let $K_{m_0,\dots,m_{d-1}}$ be a complete $d$-partite graph
and set $M := \max_{i < d} m_i$ and $N := \sum_{i < d} m_i$.
If $M \leq 2^{s(N-M+1)}$, then
\begin{align*}
  \Sep(K_{m_0,\dots,m_{d-1}},k)
  &= \begin{cases}
      1 &\text{if } k < N-M \\
      M &\text{if } k \geq N-M
    \end{cases} \\
  &\leq \begin{cases}
         2^{s(k+1)}   &\text{if } k < N-M \\
         2^{s(N-M+1)} &\text{if } k \geq N-M
       \end{cases} \\
  &\leq 2^{s(k+1)}.
\end{align*}

$(2) \Rightarrow (3)$
Suppose that $\calC$~has property $\SEP(f)$ where $f(k) = 2^{s(k+1)}$.
Note that
\begin{align*}
  \Sep(K_{m_0,\dots,m_{d-1}},k)
   = \begin{cases}
       1 &\text{if } k < N-M\,, \\
       M &\text{if } k \geq N-M\,,
     \end{cases}
\end{align*}
where $M$~and~$N$ are as above.
It follows that
\begin{align*}
  M = \Sep(K_{m_0,\dots,m_{d-1}},N-M)
    \leq f(N-M) = 2^{s(N-M+1)}.
\end{align*}
\upqed
\end{proof}

\noindent As a corollary we obtain a special case of
Conjecture~\ref{Conj: VR-equational with SEP is orderable}
for classes of complete $d$-partite graphs.
\begin{cor}\label{Cor: VR-equational d-partite graphs}
Let $\calC$~be a VR-equational class of complete $d$-partite graphs,
for some fixed natural number $d > 1$.
Then $\calC$ is $\MSO_2$-orderable if, and only if, it has property $\SEP$.
This property is decidable.
\end{cor}
\begin{proof}
For every $d \in \bbN$, there is an $\MSO$-formula
$\varphi_d(X_0,\dots,X_{d-1})$ stating that $X_0,\dots,X_{d-1}$ are
the vertex sets of a complete $d$-partite graph.
By the Semi-Linearity Theorem, it follows that the set
\begin{align*}
  M_d := \set{ (m_0,\dots,m_{d-1}) }{ K_{m_0,\dots,m_{d-1}} \in \calC }
\end{align*}
is semi-linear.

Suppose that $\calC$~has property $\SEP$.
By Example~\ref{Exam: Sep for bipartite graph}, it follows that,
for every choice of $m_0,\dots,m_{d-2}$, there are only finitely many~$m_{d-1}$
such that $K_{m_0,\dots,m_{d-2},m_{d-1}} \in \calC$.
Semi-linearity of~$M_d$ therefore implies that there are numbers $a,b \in \bbN$
such that
\begin{align*}
  m_{d-1} \leq a(m_0 +\dots+ m_{d-2}) + b\,,
  \quad\text{for all } K_{m_0,\dots,m_{d-1}} \in \calC\,.
\end{align*}
By Theorem~\ref{Thm: GSO-orderable d-partite graphs}
it follows that $\calC$~is $\MSO_2$-orderable.
\end{proof}

\subsection{Split graphs and chordal graphs}   

As the next step towards
Conjecture~\ref{Conj: VR-equational with SEP is orderable},
the case of a VR-equational class of cographs suggests itself,
but, so far, we were unable to find a proof.
(See Section~\ref{Sect: cographs} for the definition of a cograph.
Note that Corollary~\ref{Cor: VR-equational d-partite graphs}
contains a solution for complete multi-partite graphs,
which are a special kind of cographs.)
Instead, we consider split graphs and, more generally, chordal graphs.
\begin{defi}
Let $G$~be a graph.

(a)
$G$~is a \emph{split graph} if there exists a partition of its vertex set
into two parts $A$~and~$B$ such that $A$~induces a clique whereas
$B$~is independent, i.e., $G[B]$ contains no edges.

(b) Let $F$~be a spanning forest of~$G$ with tree-order~$\preceq_F$.
We call~$F$ a \emph{perfect spanning forest} if it is normal (cf.~Section~\ref{Sect: ommiting a minor})
and, for every vertex $v \in F$, the set of all neighbours~$u$ of~$v$
such that $u \prec_F v$ induces a clique in~$G$.

(c) $G$~is \emph{chordal} if it has a perfect spanning forest.
\closingmark\end{defi}

Every split graph is chordal. There are many equivalent definitions of chordal graphs.
See Proposition~2.72 of~\cite{CourcelleEngelfriet12} for an overview and a proof of their equivalence.

\begin{thm}\label{Thm: orderability of split graphs}
A class~$\calC$ of split graphs is $\MSO_2$-orderable if, and only if,
there is some $s \in \bbN$ such that
$\calC$~has property $\SEP(f)$ for the function~$f$ such that
$f(n) = 2^{s(n+1)}$.
\end{thm}
\begin{proof}
$(\Leftarrow)$
Given~$s$, we construct an $\MSO_2$-formula~$\varphi(x,y;\bar Z)$ with $s+1$ parameters that
orders every split graph~$G$ such that $\Sep(G,n) \leq 2^{s(n+1)}$, for all~$n$.
Let $G = \langle V,E\rangle$ be such a split graph and let $V = A \cup B$ be the partition of~$V$
into a clique~$A$ and an independent set~$B$.
We use one parameter~$P$ to define an order on~$A$ as follows.
Fixing an enumeration $a_0,\dots,a_{n-1}$ of~$A$ we set
\begin{align*}
  P := \{a_0\} \cup \set{ (a_i,a_{i+1}) }{ i < n-1 }\,.
\end{align*}
Then we can write down an $\MSO_2$-formula $\psi(x,y;P)$ stating that
every path that connects the unique vertex in~$P$ to~$y$
and that only uses edges in~$P$ contains the vertex~$x$.
This defines a linear order~$\leq_A$ on~$A$.

We use this order to define an order on~$B$ as follows.
For $b \in B$ let
\begin{align*}
  N(b) := \set{ a \in A }{ (a,b) \in E }\,.
\end{align*}
We first define a preorder~$\sqsubseteq$ on~$B$ by
\begin{align*}
  b \sqsubseteq b' \quad\defiff\quad
  N(b) = N(b')
  \text{ or }
  \text{the $\leq_A$-least element of } N(b) \mathbin\Delta N(b')
  \text{ belongs to } N(b)\,.
\end{align*}
Since this preorder is linear, i.e., there are no incomparable elements,
it is sufficient to define an order on each class of the equivalence relation
associated with~$\sqsubseteq$.
Given $b \in B$, we fix an enumeration $b_0,\dots,b_{m-1}$ of all vertices
$b_i \in B$ such that $N(b_i) = N(b)$ and
a $\leq_A$-increasing enumeration $a_0,\dots,a_{n-1}$ of $N(b)$.
Then
\begin{align*}
  m \leq \Sep(G,n) \leq 2^{s(n+1)}\,.
\end{align*}
Choosing an injective function $\pi : [m] \to \PSet([s(n+1)])$,
we set, for $k < s$,
\begin{align*}
  Q_k := \set{ (b_i,a_l) }{ k(n+1) + l \in \pi(i) }
    \cup \set{ b_i }{ k(n+1) + n \in \pi(i) }\,.
\end{align*}
Using the parameters $Q_0,\dots,Q_{s-1}$, we can order $b_0,\dots,b_{m-1}$ by
\begin{align*}
  b_i <_B b_j \quad\defiff\quad
  \text{the least element of } \pi(i) \mathbin\Delta \pi(j)
  \text{ belongs to } \pi(i)\,.
\end{align*}
Finally, by combining $\leq_A$, $\sqsubseteq$, and~$<_B$,
we can define an order on all vertices of~$G$.

$(\Rightarrow)$
Suppose that a split graph $G = \langle V,E\rangle$ is ordered by a formula $\varphi(x,y;\bar P)$
with $s$~parameters $P_0,\dots,P_{s-1}$. We will prove that $\Sep(G,n) \leq 2^{(s+1)(n+1)}$.
Let $V = A \cup B$ be the partition of~$V$ into a clique~$A$ and an independent set~$B$.
We start by showing that, for every $b \in B$,
there are at most $2^{s(\abs{N(b)}+1)}$ vertices $b' \in B$ with $N(b') = N(b)$,
where $N(b)$~is defined as above.
Let $b_0,\dots,b_{m-1}$ be a list of distinct vertices of~$B$
such that $N(b_0) = \dots = N(b_{m-1})$.
For a contradiction, suppose that $m > 2^{s\abs{N(b_0)}+s}$.
Then there are indices $i < j$ such that
\begin{align*}
  b_i \in P_k &\quad\iff\quad b_j \in P_k\,,         &&\quad\text{for all } k < s\,, \\
  (b_i,a) \in P_k &\quad\iff\quad (b_j,a) \in P_k\,, &&\quad\text{for all } k < s \text{ and } a \in N(b_0)\,.
\end{align*}
It follows that the mapping that interchanges $b_i$~and~$b_j$ and that
fixes every other vertex of $\langle G,\bar P\rangle$ is an automorphism. Hence,
\begin{align*}
  \lceil G\rceil \models \varphi(b_i,b_j; \bar P)
  \quad\iff\quad
  \lceil G\rceil \models \varphi(b_j,b_i; \bar P)\,,
\end{align*}
and $\varphi$~does not define an order on~$G$.
A contradiction.

To compute $\Sep(G,n)$ consider a set $S \subseteq V$ of size $\abs{S} \leq n$.
We have seen above that, for every set $X \subseteq S \cap A$, there are
at most $2^{s(\abs{X}+1)}$ vertices $b \in B$ such that $N(b) = X$.
Setting $k := \abs{S \cap A}$, it follows that there are at most
$2^k \cdot 2^{s(k+1)}$ vertices $b \in B$ such that $N(b) \subseteq S \cap A$.
Consequently, $G - S$ has at most
\begin{align*}
  1 + 2^k \cdot 2^{s(k+1)}
    \leq 2^{sk+s+k+1} = 2^{(s+1)(k+1)} \leq 2^{(s+1)(n+1)}
\end{align*}
connected components and the claim follows.
\end{proof}

\begin{lem}
For every increasing and unbounded function $g : \bbN \to \bbN$
there exists a class of split graphs that is not $\MSO_2$-orderable
but has property $\SEP(f)$ for the function~$f$ such that $f(n) := 2^{ng(n)}$.
\end{lem}
\begin{proof}
For $k \in \bbN$, let $G_k := K_k \otimes D_{2^{kg(k)}}$ where
$D_n$~denotes the graph with $n$~vertices and no edges.
We claim that $\calC := \set{ G_k }{ k \in \bbN }$ has the desired properties.
Note that
\begin{align*}
  \Sep(G_k,n) \leq \begin{cases}
                     1         &\text{if } n < k\,, \\
                     2^{ng(n)} &\text{if } n \geq k\,.
                   \end{cases}
\end{align*}
Hence, $\calC$~has property $\SEP$, but it does not have property $\SEP(f)$,
for any function~$f$ such that $f(n) = 2^{s(n+1)}$ for some $s \in \bbN$.
By Theorem~\ref{Thm: orderability of split graphs},
it follows that $\calC$~is not $\MSO_2$-orderable.
\end{proof}
\begin{rem}
The class in the preceding lemma is not VR-equational
since it does not satisfy the Semi-Linearity Theorem.
Hence, it does not provide a counterexample to
Conjecture~\ref{Conj: VR-equational with SEP is orderable}.
\closingmark\end{rem}

It would be interesting to extend Theorem~\ref{Thm: orderability of split graphs} to classes of chordal graphs.
At this point, we are only able to present a sufficient condition for $\MSO_2$-orderability.
But there are examples showing that it is not necessary.
We start with a technical lemma.
\begin{lem}\label{Lem: precedessors in perfect spanning forests}
Let $F$~be a perfect spanning forest of a chordal graph~$G$ with tree-order~$\preceq_F$.
If $u \prec_F v \preceq_F w$ are vertices then
\begin{align*}
  (u,w) \in E \qtextq{implies} (u,v) \in E\,.
\end{align*}
\end{lem}
\begin{proof}
Let $x_n \prec_F\dots\prec_F x_0$ be the path in~$F$ from $v = x_n$ to $w = x_0$.
We show by induction on~$i$, that $(u,x_i) \in E$.
For $i = 0$, there is nothing to do.
Hence, suppose that $i > 0$ and that we have already shown that $(u,x_{i-1}) \in E$.
Then $u$~and~$x_i$ are both neighbours of~$x_{i-1}$. Since $u,x_i \prec_F x_{i-1}$,
it follows by definition of a perfect spanning forest that $(u,x_i) \in E$.
\end{proof}

\begin{prop}\label{Prop: orderable chordal graphs}
Let $\calC$ be a class of chordal graphs with property $\SEP(f)$
where $f(n) = 2^{s(n+1)}$, for some $s \in \bbN$.
Then $\calC$~is $\MSO_2$-orderable.
\end{prop}
\begin{proof}
Let $G = \langle V,E\rangle$ be a chordal graph such that
$\Sep(G,n) \leq 2^{s(n+1)}$. To order~$G$,
we fix a perfect spanning forest~$F$ of~$G$.
It is sufficient to define, for every vertex~$v$, an order on the immediate successors of~$v$ in~$F$.
Then we can use the lexicographic ordering on~$F$ to order~$G$.
Fix a vertex~$v$ and let $u_0,\dots,u_{n-1}$ be the immediate successors of~$v$ in~$F$.
For $i < n$, we define
\begin{align*}
  B_i := \set{ w \preceq_F v }{ (w,u_i) \in E }\,.
\end{align*}
We start by showing that, for every set $B \subseteq V$, there are at most $2^{s(\abs{B}+1)}$
indices~$i$ such that $B_i = B$.
Given~$B$, let $I$~be the set of all $i < n$ such that $B_i = B$.
By Lemma~\ref{Lem: precedessors in perfect spanning forests},
it follows that, for every $i \in I$ and every edge $(x,y) \in E$ such that
$x \prec_F u_i \preceq_F y$, we have $x \in B_i = B$.
Hence,
\begin{align*}
  \abs{I} \leq \Sep(G,\abs{B}) \leq 2^{s(\abs{B}+1)}
\end{align*}
as desired.
As in the proof of Theorem~\ref{Thm: orderability of split graphs},
we can use $s+1$ parameters $Q_0,\dots,Q_s$ to colour the edges of the subgraphs $B_i \otimes u_i$
such a way that we can define the ordering
\begin{align*}
  u_i < u_k \quad\iff\quad i < k\,,
  \quad\text{for } i,k \in I\,.
\end{align*}
Consequently, we can order all immediate successors of~$v$ by
\begin{align*}
  u_i \leq u_k \quad\defiff\quad &B_i = B_k \text{ and } i \leq k\,, \text{ or} \\
  &\text{the $\prec_F$-least element of } B_i \mathbin\Delta B_k \text{ belongs to } B_i\,.
\end{align*}
\upqed
\end{proof}

\begin{cor}\label{Cor: orderability of VR-equational chordal graphs}
Let $\calC$~be a VR-equational class of chordal graphs.
The following statements are equivalent\?:
\begin{enumerate}
\item $\calC$~is $\MSO_2$-orderable.
\item $\calC$~has property $\SEP$.
\item There are constants $r,s \in \bbN$ such that $\calC$~has property
  $\SEP(f)$ where $f$~is the function such that $f(n) = rn+s$.
\end{enumerate}
These properties are decidable.
\end{cor}
Since we have already proved (3)~$\Rightarrow$~(1) and (1)~$\Rightarrow$~(2)
in Proposition~\ref{Prop: orderable chordal graphs} and Corollary~\ref{Cor: GSO-orderable implies SEP},
only the implication (2)~$\Rightarrow$~(3) remains to be proved.
We leave this proof to the reader\?;
it is similar to that of Corollary~\ref{Cor: VR-equational d-partite graphs}.

\section{$\MSO_1$-definable orders}   
\label{Sect: MSO}

After having studied $\MSO_2$-orderability, we consider $\MSO_1$-orderability.
For classes that are $r$-sparse, for some~$r$, $\MSO_1$~and~$\MSO_2$ have
the same expressive power (see Theorem~9.38 of~\cite{CourcelleEngelfriet12}).
For these classes we can therefore use the results of Section~\ref{Sect: MSO2}.
For general classes, $\MSO_1$-orderability turns out to be more difficult to
characterise than $\MSO_2$-orderability.

\subsection{Necessary conditions}
\label{Sect: MSO orders}

We will employ tools related to the notion of clique-width.
Instead of using the exact operations defining clique-width
(cf.~Section~\ref{Sect: structures and graphs}),
we introduce related ones that are more convenient in our context.
\begin{defi}
Let $k \in \bbN$ and $R \subseteq [k] \times [k]$.

(a)
For undirected graphs $G$~and~$H$ with ports in~$[k]$,
we construct the undirected graph $G \otimes_R H$
by adding to the disjoint union $G \oplus H$ all edges $(x,y)$ such that
\begin{itemize}
\item either $x \in G$ and $y \in H$, or $x \in H$ and $y \in G$\?; and
\item $x$~has port label~$a$ and $y$~has port label~$b$, for some $(a,b) \in R$.
\end{itemize}
Similarly, we define $G \otimes_R H$ for graphs $G$~and~$H$ with ports
expanded by additional unary predicates (vertex colours) and constants.

(b) For a graph~$G$ with ports, we denote by $\Del(G)$ the graph
obtained from~$G$ by deleting all port labels.
\closingmark\end{defi}

\begin{rem}
(a) The operation $\otimes_R$ is associative and commutative with the empty graph as neutral element.
Furthermore, ${\otimes_R} = {\otimes_{R \cup R^{-1}}}$.

(b) With only one port label, there are two operations of the
form~$\otimes_R$\?: the operations $\oplus$~and~$\otimes$ used to build
cographs (see Section~\ref{Sect: cographs} below).

(c) We have $\overline{G \otimes_R H} = \overline{G} \otimes_{R'} \overline{H}$
where $R' := ([k] \times [k]) \setminus R$ and $\overline{G}$ denotes the edge complement of~$G$.

(d) We can express $\otimes_R$ as a combination of the operations defining clique-width in the following way\?:
\begin{align*}
  G \otimes_R H = \mathrm{relab}_{h_-}(\mathrm{add}_{a_0,b_0}(\cdots\,\mathrm{add}_{a_n,b_n}(G \oplus \mathrm{relab}_{h_+}(H))\cdots))\,,
\end{align*}
for suitable functions $h_+ : [k] \to [2k]$ and $h_- : [2k] \to [k]$ and port
labels $a_0,b_0,\dots,a_n,b_n \in [2k]$. ($h_+$~is needed to make the port
labels appearing in~$H$ distinct from those appearing in~$G$.)
\closingmark\end{rem}

\begin{rem}\label{Rem: orderable closed under otimesR}
(a)
As in Proposition~\ref{Prop: orderable classes closed under union}\,(b),
one can show that
\begin{align*}
  \calC \otimes_R \calK := \set{ G \otimes_R H }{ G \in \calC,\ H \in \calK }
\end{align*}
is $\MSO$-orderable if, and only if, $\calC$~and~$\calK$ are $\MSO$-orderable.

(b) $\overline{\calC} := \set{ \overline{G} }{ G \in \calC }$ is $\MSO$-orderable if,
and only if, $\calC$~is $\MSO$-orderable.
\closingmark\end{rem}

To give a necessary condition for $\MSO_1$-orderability,
we introduce a combinatorial property
similar to $\SEP$, but based on the operation~$\otimes_R$.
\begin{defi}
Let $G$~be a graph (without port labels) and $k \in \bbN$.

(a) We denote by $\Cut(G,k)$ the maximal number~$n$ such that
there exist non\-empty graphs $H_0,\dots,H_{n-1}$ with ports in~$[k]$
and a relation $R \subseteq [k] \times [k]$
such that
\begin{align*}
  G \cong \Del(H_0 \otimes_R\dots\otimes_R H_{n-1})\,.
\end{align*}

(b) We say that a class~$\calC$ of graphs has property $\CUT(f)$,
for a function $f : \bbN \to \bbN$, if
\begin{align*}
  \Cut(G,k) \leq f(k)\,,
  \quad\text{for all } G \in \calC \text{ and all } k \in \bbN\,.
\end{align*}
We say that $\calC$~has property $\CUT$, if it has property $\CUT(f)$, for some $f : \bbN \to \bbN$.
\closingmark\end{defi}
\begin{rem}
Note that $\Cut(G,k) = \Cut(\overline{G},k)$.
\closingmark\end{rem}

For the proof that $\CUT$ is a necessary condition for $\MSO_1$-orderability,
we use the following technical lemma.
\begin{lem}\label{Lem: otimes compatible with MSO}
Let $G,G',H,H'$ be labelled graphs,
$\bar P,\bar P',\bar Q,\bar Q'$ tuples of sets of vertices of the respective graphs,
and $\bar a,\bar a',\bar b,\bar b'$ tuples of vertices.
For each port label~$c$, let $C_c,C'_c,D_c,D'_c$ be the sets of all vertices
of, respectively, $G,G',H,H'$ that have port label~$c$.
Then
\begin{align*}
  \MTh_m(\lfloor G\rfloor, \bar P,\bar C,\bar a) &=
  \MTh_m(\lfloor G'\rfloor, \bar P',\bar C',\bar a') \\
\prefixtext{and}
  \MTh_m(\lfloor H\rfloor, \bar Q,\bar D,\bar b) &=
  \MTh_m(\lfloor H'\rfloor, \bar Q',\bar D',\bar b')
\end{align*}
implies that
\begin{align*}
   \MTh_m\bigl(\lfloor G \otimes_R H\rfloor, \bar S,\bar a\bar b\bigr)
 = \MTh_m\bigl(\lfloor G' \otimes_R H'\rfloor, \bar S',\bar a'\bar b'\bigr)\,,
\end{align*}
where $S_i := P_i \cup Q_i$ and $S'_i = P'_i \cup Q'_i$.
\end{lem}
\begin{proof}
Let $\sigma$~be a quantifier-free transduction that maps a structure~$\frakA$
to its expansion $\langle\frakA,I\rangle$ where $I := A \times A$ is the
equivalence relation on~$A$ with a single class.
Given~$R$, we can write down a quantifier-free transduction~$\tau$ such that
\begin{align*}
  \bigl\langle\lfloor G \otimes_R H\rfloor, \bar S,\bar a\bar b\bigr\rangle 
  &= \tau\bigl( \sigma(\langle\lfloor G\rfloor, \bar P,\bar C,\bar a\rangle) \oplus
                \sigma(\langle\lfloor H\rfloor, \bar Q,\bar D,\bar b\rangle)\bigr) \\
\prefixtext{and}
  \bigl\langle\lfloor G' \otimes_R H'\rfloor, \bar S',\bar a'\bar b'\bigr\rangle 
  &= \tau\bigl( \sigma(\langle\lfloor G'\rfloor, \bar P',\bar C',\bar a'\rangle) \oplus
                \sigma(\langle\lfloor H'\rfloor, \bar Q',\bar D',\bar b'\rangle)\bigr)\,.
\end{align*}
This transduction uses the relation~$I$ to mark the two components of
the disjoint union.
The claim now follows from the Composition Theorem and the Backwards
Translation Lemma.
\end{proof}

\begin{prop}\label{Prop: Cut boundend}
There exists a function $f : \bbN^3 \to \bbN$ such that
$\Cut(G,k) \leq f(n,m,k)$ for every graph~$G$ such that
$\lfloor G\rfloor$ can be ordered by an $\MSO$-formula of the form
$\varphi(x,y;\bar P)$ where $\qr(\varphi) \leq m$ and
$\bar P = \langle P_0,\dots, P_{n-1}\rangle$ are parameters.
Furthermore, the function $f(n,m,k)$~is effectively elementary in the argument~$k$, that is,
there exists a computable function~$g$ such that $f(n,m,k) \leq \exp_{g(n,m)}(k)$.
\end{prop}
\begin{proof}
Fixing $k,m,n \in \bbN$, we choose for $f(n,m,k)$ an upper bound on the number
of $\MSO$-theories of the form
\begin{align*}
  \MTh_m(\lfloor H\rfloor ,v,P_0,\dots,P_{n-1},Q_0,\dots,Q_{k-1})
\end{align*}
where $H$~is a graph, $v$~is a vertex of~$H$ and $P_0,\dots,Q_0,\dots$ are
parameters.
For fixed~$m$, we can choose this bound to be elementary in~$k$.

Let $\varphi(x,y;\bar Z)$ be an $\MSO$-formula of quantifier-rank at most~$m$,
let $G$~be a graph with $\Cut(G,k) > f(n,m,k)$,
and let $P_0,\dots,P_{n-1}$ be parameters from~$G$.
We have to show that $\varphi(x,y;\bar P)$ does not order~$G$.
We choose graphs $H_0,\dots,H_{d-1}$ with $d = \Cut(G,k)$
and a relation $R \subseteq [k] \times [k]$ such that
\begin{align*}
  G = \Del(H_0 \otimes_R\dots\otimes_R H_{d-1})\,.
\end{align*}
For $c < k$, let
\begin{align*}
  C_c := \set{ x \in G }{ x \in H_i, \text{ for some } i < d, \text{ and } x \text{ has port label } c \text{ in } H_i }\,.
\end{align*}
Since $d > f(n,m,k)$, there are indices $i < j$ such that
\begin{align*}
    \MTh_m(\lfloor H_i\rfloor,a_i,\bar P \restriction {H_i},\bar C \restriction {H_i})
  = \MTh_m(\lfloor H_j\rfloor,a_j,\bar P \restriction {H_j},\bar C \restriction {H_j})\,.
\end{align*}
As there exists a graph~$F$ such that
\begin{align*}
  \langle\lfloor G\rfloor,a_ia_j,\bar P,\bar Q\rangle
           &= \langle\lfloor H_i\rfloor,a_i,\bar P \restriction {H_i},\bar C \restriction {H_i}\rangle
    \otimes_R \langle\lfloor H_j\rfloor,a_j,\bar P \restriction {H_j},\bar C \restriction {H_j}\rangle
    \otimes_R F \\
\prefixtext{and}
  \langle\lfloor G\rfloor,a_ja_i,\bar P,\bar Q\rangle
           &= \langle\lfloor H_j\rfloor,a_j,\bar P \restriction {H_j},\bar C \restriction {H_j}\rangle
    \otimes_R \langle\lfloor H_i\rfloor,a_i,\bar P \restriction {H_i},\bar C \restriction {H_i}\rangle
    \otimes_R F\,,
\end{align*}
it follows by Lemma~\ref{Lem: otimes compatible with MSO} that
\begin{align*}
  \MTh_m(\lfloor G\rfloor,a_ia_j,\bar P,\bar C) =
  \MTh_m(\lfloor G\rfloor,a_ja_i,\bar P,\bar C)\,.
\end{align*}
In particular, we have
\begin{align*}
  \lfloor G\rfloor \models \varphi(a_i,a_j;\bar P)
  \quad\iff\quad
  \lfloor G\rfloor \models \varphi(a_j,a_i;\bar P)\,.
\end{align*}
Hence, $\varphi(x,y;\bar P)$ does not define an order on~$G$.
\end{proof}

\begin{cor}\label{Cor: MSO-orderable implies CUT}
An $\MSO_1$-orderable class of graphs~$\calC$
has property $\CUT(f)$, for an elementary function~$f$.
\end{cor}
\begin{exa}
The following classes are not $\MSO_1$-orderable\?:
\begin{itemize}
\item the class of all cliques~$K_n$\?;
\item the class of all complete bipartite graphs $K_{n,m}$\?;
\item any class of graphs of the form $G \otimes (H_0 \oplus\dots\oplus H_n)$
  where the number~$n$ is unbounded and each $H_i$~is nonempty.
\end{itemize}
In each case, after fixing a number~$k$ of parameters,
we can choose a graph~$G$ that is sufficiently large
such that any colouring with~$k$ parameters
$P_0,\dots,P_{k-1}$ admits a nontrivial automorphism. Hence,
no formula can define an order on $\langle\lfloor G\rfloor,\bar P\rangle$.
\closingmark\end{exa}

As $\MSO_1$-orderability implies $\MSO_2$-orderability, we can expect that
the property $\CUT$ implies $\SEP$. The following lemma proves this fact.
\begin{lem}\label{Lem: CUT => SEP}
A class~$\calC$ of graphs with property $\CUT(f)$ has property $\SEP(g)$
where $g$~is the function such that $g(n) := f(n + 2^n) - 1$.
\end{lem}
\begin{proof}
Let $G = \langle V,E\rangle \in \calC$ and consider a set $S \subseteq V$ of
size $\abs{S} \leq n$.
Let $C_0,\dots,C_{d-1}$ be an enumeration of the connected components
of $G - S$. We claim that $d \leq g(n)$.

We define colourings $\varrho : S \to D$ and $\pi_i : C_i \to D$,
for $i < d$, as follows.
The set of colours is $D := S \cup \PSet(S)$.
(To be formally correct, we have to take the set $[k]$ where
$k := \abs{S \cup \PSet(S)}$.
To simplify notation, we will use $S \cup \PSet(S)$ instead.)
We set
\begin{align*}
  \varrho(s) := s
  \qtextq{and}
  \pi_i(v) := \set{ s \in S }{ (v,s) \in E }\,.
\end{align*}
It follows that
\begin{align*}
  G = \Del\bigl(\langle S,\varrho\rangle \otimes_R
                       \langle C_0,\pi_0\rangle \otimes_R \dots \otimes_R \langle C_{d-1},\pi_{d-1}\rangle\bigr)\,,
\end{align*}
where
\begin{align*}
  R := \set{ (s,X) \in S \times \PSet(S) }{ s \in X }\,.
\end{align*}
Consequently, $\Cut(G,\abs{D}) \geq d+1$. Since $\abs{D} \leq n+2^n$, it follows that
\begin{align*}
  d+1 \leq \Cut(G,n+2^n) \leq f(n+2^n) = g(n)+1\,.
\end{align*}
\upqed
\end{proof}

The converse obviously does not hold.
A special case, where it \emph{does} hold is the case of $r$-sparse graphs (cf.\ Definition~\ref{Def: sparse}).
This case is of particular interest since, for $r$-sparse graphs, the expressive powers
of $\MSO_1$ and $\MSO_2$ coincide (see Theorem~9.37 of~\cite{CourcelleEngelfriet12}).
\begin{lem}\label{Lem: complete bipartite graphs are not-sparse}
The graph $K_{m,n}$ is $r$-sparse if, and only if,
$r \geq \frac{mn}{m+n}$.
\end{lem}
\begin{proof}
Every induced subgraph of~$K_{m,n}$ is of the form $K_{m',n'}$
with $m' \leq m$ and $n' \leq n$.
Such a subgraph has $m' + n'$ vertices and $m'n'$ edges.
The ratio is
\begin{align*}
  \frac{m'n'}{m'+n'} = \frac{1}{\frac{1}{m'} + \frac{1}{n'}}
    \leq \frac{1}{\frac{1}{m} + \frac{1}{n}} = \frac{mn}{m+n}\,.
\end{align*}
\upqed
\end{proof}

\begin{lem}\label{Lem: sparse SEP => CUT}
A class~$\calC$ of $r$-sparse graphs with property $\SEP(f)$
has property $\CUT(g)$ where $g(k) := f(2k^2r(2r+1))$.
\end{lem}
\begin{proof}
Let $G \in \calC$.
Suppose that
\begin{align*}
  G = \Del\bigl((H_0,\pi_0) \otimes_R\dots\otimes_R (H_{d-1},\pi_{d-1})\bigr)
  \quad\text{where } R \subseteq [k] \times [k]\,.
\end{align*}
Without loss of generality,
we may assume that $R$~is symmetric.
We have to show that $d \leq g(k)$.

Set $I_a := \set{ i < d }{ \pi_i^{-1}(a) \neq \emptyset }$.
First, let us show that
\begin{align*}
  \abs{I_a} \leq 2r+1
  \qtextq{or}
  \abs{I_b} \leq 2r+1\,,
  \qquad\text{for every } (a,b) \in R\,.
\end{align*}
For a contradiction, suppose that there is some $(a,b) \in R$
that $\abs{I_a} \geq 2r+2$ and $\abs{I_b} \geq 2r+2$.
Choose subsets $I'_a \subseteq I_a$ and $I'_b \subseteq I_b$
of size $m := 2r+2$ and select vertices $x_i \in \pi_i^{-1}(a)$, for $i \in I'_a$,
and $y_i \in \pi_i^{-1}(b)$, for $i \in I'_b$.
The subgraph induced by these vertices has $m^2 - \abs{I_a \cap I_b} \geq m^2 - m$ edges and $2m$ vertices.
Since
\begin{align*}
  \frac{m^2-m}{2m} = \frac{m-1}{2} = \frac{2r+1}{2} > r\,,
\end{align*}
it follows that $G$~is not $r$-sparse. A contradiction.

For $a,b \in [k]$, we set
\begin{align*}
  S_{ab} &:= \bigcup {\bigset{ \pi_i^{-1}(a) }{ \textstyle i \in I_a,\ \abs{\pi_i^{-1}(a)} \leq 2r }}\,, \\
       S &:= \bigcup {\bigset{ S_{ab} }{ (a,b) \in R,\ \abs{I_a} \leq 2r+1 }}\,.
\end{align*}
Note that
\begin{align*}
  \abs{S_{ab}} \leq 2r\abs{I_a}
  \qtextq{and}
  \abs{S} \leq \abs{R}\cdot(2r+1)\cdot(2r) \leq 2k^2r(2r+1)\,.
\end{align*}

We claim that every connected component of $G - S$ is contained in
$H_i - S$, for some~$i$.
For a contradiction, suppose that there is a connected component~$C$
of $G - S$ containing vertices from both $H_i - S$ and $H_j - S$.
Then there exists an edge $(x,y)$ of~$G$ with $x \in H_i - S$
and $y \in H_j - S$. Let $a := \pi_i(x)$ and $b := \pi_j(y)$.
Then $(a,b) \in R$.
We have shown above that $\abs{I_a} \leq 2r+1$ or $\abs{I_b} \leq 2r+1$.
In the first case, we have $x \in \pi_i^{-1}(a) \subseteq S_{ab} \subseteq S$,
in the second case, we have $y \in \pi_i^{-1}(b) \subseteq S_{ba} \subseteq S$.
Hence, both cases lead to a contradiction.

It follows that $G - S$ has at least~$d$ connected components.
Consequently,
\begin{align*}
  d \leq \Sep(G,\abs{S}) \leq \Sep(G,2k^2r(2r+1)) \leq f(2k^2r(2r+1)) = g(k)\,.
\end{align*}
\upqed
\end{proof}

\subsection{Cographs}   
\label{Sect: cographs}

A well-known VR-equational class is the class of cographs.
A \emph{cograph} is a graph that can be constructed from single vertices
using the operations of disjoint union~$\oplus$ and complete join~$\otimes$.
Each cograph can be denoted by a term over $\oplus$,~$\otimes$,
and a constant~$1$ that denotes an isolated vertex.
For instance, $(1 \oplus 1) \otimes (1 \oplus 1 \oplus 1)$ denotes the graph
$K_{2,3}$, and $1 \otimes 1 \otimes\dots\otimes 1$ denotes a clique.
Since $\oplus$~and~$\otimes$ are associative and commutative, we consider them
as operations of variable arity and we ignore the order of the arguments.
The class~$\calC$ of cographs is VR-equational.
It can be defined by the equation
\begin{align*}
  \calC = \calC \oplus \calC \cup \calC \otimes \calC \cup \{1\}\,.
\end{align*}

A cograph~$G$ with more than one vertex is either disconnected and of the
form $G = H_0 \oplus\dots\oplus H_n$ for connected cographs $H_0,\dots,H_n$,
or it is connected and of the form $G = H_0 \otimes\dots\otimes H_n$ for
cographs $H_0,\dots,H_n$ each of which is either disconnected or
a single vertex.
Furthermore, these decompositions of~$G$ are unique, up to the ordering of
$H_0,\dots,H_n$.
Using this observation, we can associate with every cograph a unique term
as follows.
\begin{defi}
A term~$t$ over the operations $\oplus$,~$\otimes$,~$1$
(where we consider $\oplus$~and~$\otimes$ as many-ary operations with
unordered arguments) is a \emph{cotree}
if there is no node that is labelled by the same operation as one of its
immediate successors.
Every cograph has a unique cotree.
The \emph{depth} of a cograph is the height of this cotree.
\closingmark\end{defi}
\begin{exa}
The cograph~$G$ defined by the term
\begin{align*}
  (1 \otimes (1 \oplus (1 \oplus (1 \otimes 1)))) \otimes ((1 \otimes (1 \otimes 1)) \oplus 1)
\end{align*}
has the cotree
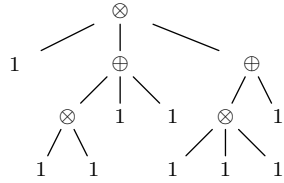
\begin{center}
%
%
%
%
\includegraphics{Order-final.2}
\end{center}
The leaves of this tree correspond to the vertices of~$G$ and
every subtree is the cotree of an induced subgraph of~$G$.
\closingmark\end{exa}

Recall (see, e.g., \cite{Courcelle96b}) that a \emph{module} of a graph
$G = \langle V,E\rangle$ is a set~$M$ of vertices such that
every vertex in $V \setminus M$ is either adjacent to all elements of~$M$,
or to none of them.
A module~$M$ is called \emph{strong} if there is no module~$N$ such that
$M \setminus N$ and $N \setminus M$ are both nonempty
(cf.~\cite{MohringRadermacher84,Courcelle96b,EhrenfeuchtHarjuRozenberg99}).
Clearly, being a module and being a strong module are expressible in $\MSO_1$.
In a cograph there are two types of strong modules\?:
the connected and the disconnected ones.

\begin{thm}\label{Thm: cographs}
Let $\calC$~be a class of cographs. The following statements are equivalent.
\begin{enumerate}
\item $\calC$~is $\MSO_1$-orderable.
\item $\calC$~has property $\CUT$.
\item There exists a constant $d \in \bbN$ such that the cotree of every
  graph in~$\calC$ has outdegree at most~$d$.
\end{enumerate}
\end{thm}
\begin{proof}
$(3) \Rightarrow (1)$ is Corollary~6.12 from~\cite{Courcelle96b}
and $(1) \Rightarrow (2)$ was shown in Corollary~\ref{Cor: MSO-orderable implies CUT}.

For $(2) \Rightarrow (3)$, suppose that, for every $d \in \bbN$,
there exists a graph~$G_d \in \calC$ with a cotree of maximal outdegree at least~$d$.
It is sufficient to show that $\Cut(G_d,3) > d$.

By assumption, we can find a strong module~$A$ of~$G_d$ containing strong submodules
$B_0,\dots,B_{n-1}$, for $n > d$, such that
either (i) $A = B_0 \oplus\dots\oplus B_{n-1}$, or (ii) $A = B_0 \otimes\dots\otimes B_{n-1}$.
Let $C := G - A$ be the graph induced by the complement of~$A$.
Every vertex $v \in C$ is either connected
to all vertices of~$A$, or to none of them.
We assign the port label~$0$ to the former vertices and the port label~$1$ to the latter ones.
Each vertex of~$A$ gets port label~$2$.
It follows that
\begin{align*}
  G_d = C \otimes_R B_0 \otimes_R\dots\otimes B_{n-1}
\end{align*}
where $R = \{(0,2),(2,0)\}$ or $R = \{(0,2),(2,0),(2,2)\}$.
Consequently, we have $\Cut(G_d,3) \geq n+1 > d$.
\end{proof}

\begin{cor}\label{Cor: cographs unorderable}
Let $k \in \bbN$. The class of cographs of depth at most~$k$ is hereditarily $\MSO_1$-unorderable.
\end{cor}
\begin{proof}
For any given depth~$k$, there are only finitely many cographs (up to isomorphism) satisfying
condition~(3) of Theorem~\ref{Thm: cographs}.
\end{proof}

\begin{cor}
For VR-equational classes of cographs, $\MSO_1$-orderability is decidable.
\end{cor}
\begin{proof}
Let $\calC$~be a VR-equational class of cographs.
By Theorem~\ref{Thm: cographs}, it is sufficient to decide whether there is a constant~$d$
such that every cotree of a graph in~$\calC$ has maximal outdegree at most~$d$.
Let $\varphi(X)$ be an $\MSO_1$-formula stating that there exists a strong module~$Z$ such that
$X \subseteq Z$ and
every strong module $Y \subset Z$ contains at most one element of~$X$.
Given a cograph~$G$, it follows that the maximal outdegree of the cotree of~$G$
is equal to the maximal size of a set~$X$ satisfying~$\varphi$ in~$G$.
Using the Semi-Linearity Theorem, we can decide whether this size is bounded.
\end{proof}

\begin{rem}
If a class~$\calC$ of cographs is $\MSO_1$-orderable,
there exists an $\MSO$-transduction
mapping each graph in~$\calC$ to its cotree (see~\cite{Courcelle96b}).
But, conversely, the existence of such an $\MSO$-transduction is not enough to ensure $\MSO_1$-orderability\?:
there exists an $\MSO$-transduction from the class of all cographs of depth~$k$
to their respective cotrees (this is a routine construction).
But, as we have just seen, this class is hereditarily $\MSO_1$-unorderable.
\closingmark\end{rem}

\subsection{$\otimes$-decompositions}   

Cographs are precisely the graphs of clique-width~$2$.
A natural aim is thus to extend the equivalence (1)~$\Leftrightarrow$~(2) of Theorem~\ref{Thm: cographs}
to classes of graphs of bounded clique-width.
However, we must leave this as a conjecture.
Instead we only consider the special case of graphs where the height of
the decomposition (as defined below) is bounded.
Such graphs generalise cographs of bounded depth, and we show that they
are hereditarily $\MSO_1$-unorderable.

We start by introducing a kind of decomposition associated with the notion of
clique-width.
\begin{defi}
Let $G = \langle V,E\rangle $ be a graph.

(a) A \emph{$\otimes$-decomposition} of~$G$ of \emph{width~$k$} is a family
$(H_v)_{v \in T}$ of labelled graphs $H_v = \langle U_v,F_v,\pi_v\rangle$
with $\pi_v : U_v \to [k]$ such that
\begin{itemize}
\item the index set~$T$ is a rooted tree,
\item $H_{\emptyseq} = \langle V,E,\pi_{\emptyseq}\rangle$, for some labelling $\pi_{\emptyseq}$,
\item $\abs{U_v} = 1$, for every leaf $v \in T$,
\item for every internal node $v \in T$ with immediate successors
  $u_0,\dots,u_{d-1}$, there is some $R_v \subseteq [k] \times [k]$ such that
  \begin{align*}
    \Del(H_v) = \Del(H_{u_0} \otimes_{R_v}\dots\otimes_{R_v} H_{u_{d-1}})\,.
  \end{align*}
\end{itemize}
We call $\otimes_{R_v}$ the \emph{operation at}~$v$.
Note that the port labels of~$H_v$ and $H_{u_0},\dots,H_{u_{d-1}}$ are
unrelated.
(Hence, the labelling~$\pi_{\emptyseq}$ of the root is arbitrary.
We have added it to keep the notation uniform.)

(b) A \emph{strong $\otimes$-decomposition} of~$G$ is
a $\otimes$-decomposition $(H_v)_{v \in T}$
such that, for each internal node $v \in T$ with immediate successors
$u_0,\dots,u_{d-1}$,
there is some $R_v \subseteq [k] \times [k]$ and some function
$\varrho : [k] \to [k]$ such that
\begin{align*}
  H_v = \mathrm{relab}_\varrho(H_{u_0} \otimes_{R_v}\dots\otimes_{R_v} H_{u_{d-1}})\,.
\end{align*}

(c) The \emph{height} of a $\otimes$-decomposition $(H_v)_{v \in T}$
is the height of the tree~$T$.

(d) We define $\WD^\otimes_n(G)$ as the least number~$k$ such that
$G$~has a $\otimes$-decom\-pos\-i\-tion of width at most~$k$ and height at most~$n$.
Similarly, we define $\SWD^\otimes_n(G)$ as the least number~$k$ such that
$G$~has a strong $\otimes$-decomposition of width at most~$k$ and height at most~$n$.
We call $\WD^\otimes_n(G)$ the \emph{$n$-depth $\otimes$-width} of~$G$
and $\SWD^\otimes_n(G)$ is its \emph{strong $n$-depth $\otimes$-width.}%
\footnote{Recently a closely related notion, called
\emph{shrub-depth,} was introduced in \cite{GanianHNOR12}. Its exact relation
to strong $n$-depth $\otimes$-width remains to be investigated.}
\closingmark\end{defi}

\begin{rem}
(a) For every graph~$G$ and all $n$,~$m$ such that $m < n$, we have
\begin{align*}
   \WD^\otimes_n(G) &\leq \SWD^\otimes_n(G) \leq \abs{V}\,, \\
   \WD^\otimes_n(G) &\leq \WD^\otimes_m(G)\,, \\
\prefixtext{and}
  \SWD^\otimes_n(G) &\leq \SWD^\otimes_m(G)\,.
\end{align*}

(b) Recall the definition of clique-width in
Section~\ref{Sect: structures and graphs}.
Since the operation~$\otimes_R$ can be expressed by the operations
clique-width is based on, but by using twice as many port labels,
it follows that the clique-width of a graph is at most twice
its strong $n$-depth $\otimes$-width (for any~$n$).
Since, conversely, for sufficiently large~$n$,
the strong $n$-depth $\otimes$-width of a graph~$G$ is at most its clique-width,
it follows that, for every graph~$G$ and all sufficiently large $n$,
\begin{align*}
  \SWD^\otimes_n(G) \leq \cwd(G) \leq 2\cdot\SWD^\otimes_n(G)\,.
\end{align*}
If we define $\SWD^\otimes(G)$ as the minimal value of $\SWD^\otimes_n(G)$
when $n$~ranges over~$\bbN$,
we therefore obtain a nontrivial width measure that is equivalent to
clique-width.

(c) Note that $\WD^\otimes_n(G) \leq 2$, for every graph~$G$ with $n$~vertices.
Hence, the width $\WD^\otimes_n(G)$ is only of interest if there is a bound on~$n$.
\closingmark\end{rem}
Because of its relation to clique-width, the strong $\otimes$-width
is of more interest than the $\otimes$-width (which becomes trivial for large depths).
We have introduced the simpler notion of $\otimes$-width since,
in the special case we consider, there exists a bound on the depth of $\otimes$-decompositions.
In this case we can use the following lemma to transform
a bound on the $\otimes$-width of a class into
a bound on its strong $\otimes$-width.
\begin{lem}
For every graph~$G$ and every $n \in \bbN$,
\begin{align*}
  \WD^\otimes_n(G) \leq \SWD^\otimes_n(G) \leq \bigl[\WD^\otimes_n(G)\bigr]^{n+1}.
\end{align*}
\end{lem}
\begin{proof}
The first inequality being trivial, we only prove the second one.
Given a $\otimes$-decom\-pos\-i\-tion $(H_v)_{v \in T}$ of~$G$ of height~$n$
and width $k := \WD^\otimes_n(G)$,
we construct a strong $\otimes$-decom\-pos\-i\-tion $(H'_v)_{v \in T}$
of~$G$ of the same height and width $k^n$.
Consider $v \in T$ and let $v_0,\dots,v_m$ be the path in~$T$ from
the root $\emptyseq = v_0$ to $v = v_m$, where $m < n$.
Suppose that $H_v = \langle U_v,F_v,\pi_v\rangle$.
We set $H'_v := \langle U_v,F_v,\pi'_v\rangle$ where
\begin{align*}
  \pi'_v(x) := \langle\pi_{v_0}(x),\dots,\pi_{v_m}(x)\rangle\,.
\end{align*}
This labelling uses $1 + k + k^2 +\dots+ k^n \leq k^{n+1}$ port labels.
Then
\begin{align*}
  H'_v = \mathrm{relab}_\varrho(H'_{u_0} \otimes_{R_v}\dots\otimes_{R_v} H'_{u_{d-1}})\,,
\end{align*}
where the function~$\varrho$ maps $\langle a_0,\dots,a_m,a_{m+1}\rangle$
to $\langle a_0,\dots,a_m\rangle$.
\end{proof}

\begin{lem}\label{Lem: CUT implies bounded degree of decomposition}
Let $G$~be a graph and $(H_v)_{v \in T}$ a $\otimes$-decomposition of~$G$
of width at most~$k$.
Every vertex of~$T$ has less than $\Cut(G,k+2^k)$ immediate successors.
\end{lem}
\begin{proof}
Suppose that $H_v = \langle U_v,F_v,\pi_v\rangle$.
Let $v \in T$ be a vertex with immediate successors $u_0,\dots,u_{m-1}$.
Hence,
\begin{align*}
  H_v = H_{u_0} \otimes_R\dots\otimes_R H_{u_{m-1}}\,,
\end{align*}
where $\otimes_R$~is the operation at~$v$.
Let $C := G - H_v$, i.e., the subgraph induced by the complement of the set of vertices of~$H_v$.
We claim that
\begin{align*}
  G = C \otimes_{R'} H_{u_0} \otimes_{R'}\dots\otimes_{R'} H_{u_{m-1}}\,,
\end{align*}
for a suitable labelling $\varrho : C \to [k+2^k]$ of~$C$ and
a suitable relation $R' \subseteq [k+2^k] \times [k+2^k]$.
This implies that $m+1 \leq \Cut(G,k+2^k)$, as desired.

It remains to define $\varrho$~and~$R'$.
Fix a bijection $\pi_0 : \PSet([k]) \to [2^k]$ and set $\pi(B) := \pi_0(B) + k$, for $B \subseteq [k]$.
Defining
\begin{align*}
  \varrho(x) &:= \pi(\set{ \pi_v(y) }{ y \in U_v\,,\ (x,y) \in E })\,, \qquad\text{for } x \in C\,,\\[1ex]
  \prefixtext{and}
  R' &:= R \cup \set{ (a,\pi(B)) }{ a \in [k],\ B \subseteq [k],\ a \in B }\,,
\end{align*}
we obtain
$G = C \otimes_{R'} H_{u_0} \otimes_{R'}\dots\otimes_{R'} H_{u_{m-1}}$.
\end{proof}

We obtain the following characterisation
of $\MSO_1$-orderable classes of bounded $n$-depth $\otimes$-width.
\begin{thm}\label{Thm: characterisation of orderable graphs of bounded otimes-width}
Let $\calC$~be a class of graphs such that, for some $n,k \in \bbN$,
\begin{align*}
  \WD^\otimes_n(G) \leq k\,,
  \quad\text{for all } G \in \calC\,.
\end{align*}
The following statements are equivalent\?:
\begin{enumerate}
\item $\calC$~is $\MSO_1$-orderable.
\item $\calC$~has property $\CUT$.
\item There is a constant $d \in \bbN$ such that every $G \in \calC$ has a
  $\otimes$-de\-com\-pos\-i\-tion $(H_v)_{v \in T}$
  of height at most~$n$ and width at most~$k$ where every vertex of~$T$
  has outdegree at most~$d$.
\item $\calC$~is finite.
\end{enumerate}
\end{thm}
\begin{proof}
(4)~$\Rightarrow$~(1) is trivial and
(1)~$\Rightarrow$~(2) follows from Corollary~\ref{Cor: MSO-orderable implies CUT}.

(2)~$\Rightarrow$~(3)
Suppose that $\calC$~has property $\CUT(f)$.
Let $G \in \calC$ and let $(H_v)_{v \in T}$ be a
$\otimes$-decomposition of~$G$ of height at most~$n$ and width at most~$k$.
Then it follows by Lemma~\ref{Lem: CUT implies bounded degree of decomposition}
that every vertex of~$T$ has less than $d := f(k+2^k)$ immediate successors.

(3)~$\Rightarrow$~(4) Since every tree of height at most~$n$
and maximal outdegree at most~$d$ has at most
$1 + (d-1) + (d-1)^2 +\dots+ (d-1)^{n-1} < d^n$
vertices, it follows that every graph in~$\calC$ has at most that many elements.
\end{proof}

We obtain the following extension of Corollary~\ref{Cor: cographs unorderable}.
\begin{cor}
For every $n,k \in \bbN$,
the class of all graphs of $n$-depth $\otimes$-width at most~$k$ is
hereditarily $\MSO_1$-unorderable.
\end{cor}

\section{Reductions between difficult cases}   
\label{Sect: reductions}

In this section we consider classes of graphs for which the question of
orderability is as hard as in the general case.
\begin{defi}
Let $G = \langle V,E\rangle$ be a graph.

(a) The \emph{incidence graph} of~$G$ is the graph $\Inc(G) := \langle V \cup E, I,P\rangle$
where the edge relation
\begin{align*}
  I := \mathrm{inc} \cup \mathrm{inc}^{-1}
    = \set{ (x,y) }{ x \text{ is an end-vertex of } y \text{ or } y \text{ is an end-vertex of } x }
\end{align*}
is the symmetric version of the incidence relation
and $P := V$ is a unary relation identifying the vertices of~$G$.

(b) The \emph{incidence split graph} of~$G$ is the graph $\IS(G) := \langle V \cup E, J\rangle$
where
\begin{align*}
  J := I \cup \set{ (x,y) \in V \times V }{ x \neq y }
\end{align*}
and $I$~is the symmetric incidence relation from~(a).
Note that $\IS(G)$ is a split graph.

(c) For a class of graphs~$\calC$, we set
\begin{align*}
  \Inc(\calC) := \set{ \Inc(G) }{ G \in \calC }
  \qtextq{and}
  \IS(\calC) := \set{ \IS(G) }{ G \in \calC }\,.
\end{align*}
\upqed
\closingmark\end{defi}

The proposition below suggests that obtaining
a characterisation of $\MSO_1$-orderability for classes of split graphs is as
hard as obtaining one of $\MSO_2$-orderability for arbitrary classes of graphs.
We start with a technical lemma.
\begin{lem}
Let $\calC$~be a class of graphs.
\begin{enumerate}[label=\upshape(\alph*)]
\item $\calC$~has property $\SEP$ if, and only if, $\Inc(\calC)$ has property $\SEP$.
\item $\Inc(\calC)$~has property $\CUT$ if, and only if, $\IS(\calC)$ has property $\CUT$.
\end{enumerate}
\end{lem}
\begin{proof}
(a) $(\Leftarrow)$
Suppose that $\Inc(\calC)$~has property $\SEP(f)$, for some $f : \bbN \to \bbN$.
We claim that $\calC$~also has property $\SEP(f)$.
Let $G = \langle V,E\rangle$ be a graph in~$\calC$.
To compute $\Sep(G,k)$ consider a set $S \subseteq V$ of cardinality
$\abs{S} \leq k$.
Let $C_0,\dots,C_{m-1}$ be the connected components of $G - S$.
Then the connected components of $\Inc(G) - S$ are
$C'_0,\dots,C'_{m-1},e_0,\dots,e_{n-1}$
where $e_0,\dots,e_{n-1}$ is an enumeration of the edges of~$G[S]$
and $C'_i$~is the induced subgraph of~$\Inc(G)$ that is obtained from
$\Inc(C_i)$ by adding (as vertices) all edges of~$G$
connecting a vertex in~$S$ to some vertex of~$C_i$.
It follows that
\begin{align*}
  \Sep(G,k) \leq \Sep(\Inc(G),k) \leq f(k)\,.
\end{align*}

$(\Rightarrow)$
Suppose that $\calC$~has property $\SEP(f)$, for some $f : \bbN \to \bbN$.
Let $G = \langle V,E\rangle$ be a graph in~$\calC$ with
$\Inc(G) = \langle V \cup E,I,P\rangle$.
To compute $\Sep(\Inc(G),k)$ we consider a set $S \subseteq V \cup E$ of size
$\abs{S} \leq k$.
For each edge $e \in S \cap E$, we select one end-vertex.
Let $X$~be the set of these end-vertices and set $S' := (S \setminus E) \cup X$.
Then $\Inc(G) - S'$ has at least as many connected components as $\Inc(G) - S$.
Since $S' \subseteq V$ it follows by what we have seen above that
$\Inc(G) - S'$ has at most $m + \binom{k}{2}$ connected components, where $m$~is the
number of connected components of $G - S'$. Consequently,
\begin{align*}
  \Sep(\Inc(G),k) \leq \Sep(G,k) + \frac{k}{2}(k-1)\,.
\end{align*}
It follows that $\Inc(\calC)$ has property $\SEP(f')$ for the function~$f'$
such that $f'(k) = f(k) + \frac{k}{2}(k-1)$.

(b) $(\Rightarrow)$
Suppose that $\Inc(\calC)$~has property $\CUT(f)$, for some $f : \bbN \to \bbN$.
Let $\Inc(G) = \langle V \cup E,I,P\rangle$ be a graph in $\Inc(\calC)$
and let $\IS(G) = \langle V \cup E,J\rangle$.
To compute $\Cut(\IS(G),k)$ suppose that
\begin{align*}
  \IS(G) = \Del(H_0 \otimes_R \dots\otimes_R H_{m-1})\,,
\end{align*}
for $k$-labelled graphs $H_0,\dots,H_{m-1}$ and a relation $R \subseteq [k] \times [k]$.
Suppose that $H_i = \langle U_i,J_i\rangle$, for $i < m$, and let
$\pi_i$~be the labelling of~$H_i$.
We set $H'_i := \langle U_i,I_i,P_i\rangle$ where $I_i := J_i \setminus (V \times V)$
and $P_i := U_i \cap V$. We label $H'_i$ by
\begin{align*}
  \pi'_i(v) := \begin{cases}
                  \pi_i(v)   &\text{if } v \notin V\,, \\
                  \pi_i(v)+k &\text{if } v \in V\,.
                \end{cases}
\end{align*}
Then $\Inc(G) = \Del(H'_0 \otimes_{R'} \dots\otimes_{R'} H'_{m-1})$, where
\begin{align*}
  R' := \set{ (x,y),(x+k,y),(x,y+k) }{ (x,y) \in R }\,.
\end{align*}
Consequently, $\Cut(\IS(G),k) \leq \Cut(\Inc(G),2k) \leq f(2k)$.

$(\Leftarrow)$
Suppose that $\IS(\calC)$~has property $\CUT(f)$, for some $f : \bbN \to \bbN$.
Let $\Inc(G) = \langle V \cup E,I,P\rangle$ be a graph in $\Inc(\calC)$
and let $\IS(G) = \langle V \cup E,J\rangle$.
To compute $\Cut(\Inc(G),k)$ suppose that
\begin{align*}
  \Inc(G) = \Del(H_0 \otimes_R \dots\otimes_R H_{m-1})\,,
\end{align*}
for $k$-labelled graphs $H_0,\dots,H_{m-1}$ and a relation $R \subseteq [k] \times [k]$.
Suppose that $H_i = \langle U_i,I_i,P_i\rangle$, for $i < m$, and let
$\pi_i$~be the labelling of~$H_i$.
We define the graph $H'_i := \langle U_i,J_i\rangle$
where $J_i := I_i \cup \set{ (x,y) }{ x,y \in P_i,\ x \neq y }$
with labelling
\begin{align*}
  \pi'_i(v) := \begin{cases}
                  \pi_i(v)     &\text{if } v \in V\,, \\
                  \pi_i(v) + k &\text{if } v \notin V\,.
                \end{cases}
\end{align*}
Then $\IS(G) = \Del(H'_0 \otimes_{R'} \dots\otimes_{R'} H'_{m-1})$, where
\begin{align*}
  R' := ([k] \times [k]) \cup
        \set{ (x,y),(x+k,y),(x,y+k),(x+k,y+k) }{ (x,y) \in R }\,.
\end{align*}
Consequently, $\Cut(\Inc(G),k) \leq \Cut(\IS(G),2k) \leq f(2k)$.
\end{proof}

\begin{prop}
Let $\calC$~be a class of graphs.
\begin{enumerate}[label=\upshape(\alph*)]
\item $\calC$~is $\MSO_2$-orderable if, and only if, $\IS(\calC)$ is $\MSO_1$-orderable.
\item $\calC$~has property $\SEP$ if, and only if, $\IS(\calC)$ has property $\CUT$.
\end{enumerate}
\end{prop}
\begin{proof}
(a) is a routine construction.
(b) follows by the preceding lemma since $\Inc(\calC)$ is $2$-sparse
and, by Lemmas \ref{Lem: CUT => SEP}~and~\ref{Lem: sparse SEP => CUT},
such a class has property $\SEP$ if, and only if, it has property $\CUT$.
\end{proof}

\begin{cor}
Let $\calP$~be a graph property such that a class of split graphs is $\MSO_1$-orderable
if, and only if, it has properties $\CUT$ and $\calP$.
Then a class of arbitrary graphs is $\MSO_2$-orderable if, and only if,
it has properties $\SEP$ and $\IS^{-1}(\calP)$.
\end{cor}
\begin{rem}
(a) Characterising $\MSO_2$-orderable classes therefore amounts to
characterising $\MSO_1$-orderable classes of split graphs contained
in the image of the function~$\IS$.

(b) If $\calC$~is a class of graphs with property $\SEP$ that is not $\MSO_2$-orderable,
then $\IS(\calC)$ is a class of split graphs with property $\CUT$ that is not $\MSO_1$-orderable.
\closingmark\end{rem}

We also present a lemma suggesting that finding a characterisation of
$\MSO_1$-order\-ability for classes of bipartite graphs is as hard as
finding a characterisation of $\MSO_1$-order\-ability
for arbitrary classes of graphs. We leave the proof -- which is similar to the one above -- to the reader.
\begin{defi}
For a graph $G = \langle V,E\rangle$ we define
\begin{align*}
  \mathrm{BP}(G) := \langle V \times [4], E'\rangle
\end{align*}
where
\begin{align*}
  E' := {} &\bigset{ ((x,0),(y,3)) }{ (x,y) \in E } \\
{} \cup {} &\bigset{ ((x,i),(x,i+1)) }{ x \in V,\ 0 \leq i < 3 }\,.
\end{align*}
For classes~$\calC$ of graphs, we define
$\mathrm{BP}(\calC) := \set{ \mathrm{BP}(G) }{ G \in \calC }$ as usual.
\closingmark\end{defi}
\begin{lem}
Let $\calC$ be a class of graphs.
\begin{enumerate}[label=\upshape(\alph*)]
\item $\calC$~is $\MSO_1$-orderable if, and only if, $\mathrm{BP}(\calC)$ is $\MSO_1$-orderable.
\item $\calC$~has property $\CUT$ if, and only if, $\mathrm{BP}(\calC)$ has property $\CUT$.
\end{enumerate}
\end{lem}

\section{Conclusion}   

For arbitrary classes of graphs, it is difficult to obtain necessary and sufficient
conditions for $\MSO_i$-orderability, as there are many different ways to
construct $\MSO$-definable orderings depending on many different structural
properties of the considered graphs.
General conditions should thus cover simultaneously a large number of
possibilities.
It is therefore necessary to consider particular graph classes.
We have obtained necessary and sufficient conditions in Theorems
\ref{Thm: GSO-orderable with excluded minor},
\ref{Thm: GSO-orderable d-partite graphs},
\ref{Thm: orderability of split graphs}, and~\ref{Thm: cographs}
with corresponding decidability results for the VR-equational classes of graphs.

Concerning future work, we think that the following questions should be
fruitfully investigated\?:

(a) Does Conjecture~\ref{Conj: VR-equational with SEP is orderable} hold\??
We have already proved several special cases and more cases seem to be within
reach.
It remains to be seen whether the full conjecture can be solved.

(b) Which condition must be added to the property $\SEP$ to yield a necessary
and sufficient condition for $\MSO_2$-orderability of a class of cographs\??
And more generally, for graph classes of bounded clique-width\??

(c) What could be an extension of Theorem~\ref{Thm: cographs}, say,
for classes of `bounded strong $\otimes$-width'\??

(d) Which operations do preserve $\MSO_i$-orderability\?? Candidates include
the operations defining tree-width or clique-width, graph substitutions, and
monadic second-order transductions.
We presented a few simple results in
Proposition~\ref{Prop: orderable classes closed under union} and
Remark~\ref{Rem: orderable closed under otimesR},
but it should not be too hard to develop a more comprehensive theory.

\end{document}